\documentclass[twocolumns,9pt]{IEEEtran}
\usepackage{amsmath}
\usepackage{amsfonts}
\usepackage{mathrsfs}
\usepackage{pifont}
\usepackage{amssymb}
\usepackage{verbatim}
\usepackage{upgreek}
\usepackage{color}
\usepackage{graphicx}
\usepackage{algorithmic}
\usepackage{algorithm}
\usepackage{epsfig}

\usepackage{fancybox}

\usepackage{amsmath,epsfig,amsfonts,amssymb,graphics,theorem,calc,url,bm}
\usepackage{array}
\usepackage{calc}
\usepackage{verbatim}

\usepackage{amsfonts}
\usepackage{mathrsfs}
\usepackage{pifont}

\usepackage{verbatim}
\usepackage{upgreek}
\usepackage{color}

\usepackage{algorithmic}
\usepackage{algorithm}

\usepackage{fancybox}

\newcommand{\bzero}{\mbox{\boldmath{$0$}}}

\newcommand{\bA}{\mbox{\boldmath{$A$}}}
\newcommand{\ba}{\mbox{\boldmath{$a$}}}
\newcommand{\bB}{\mbox{\boldmath{$B$}}}

\newcommand{\bI}{\mbox{\boldmath{$I$}}}
\newcommand{\ii}{\mbox{\boldmath $i$}}%%%%%%%%%%%%%%%%%%

\newcommand{\bn}{\mbox{\boldmath{$n$}}}

\newcommand{\bQ}{\mbox{\boldmath{$Q$}}}

\newcommand{\bR}{\mbox{\boldmath{$R$}}}

\newcommand{\bs}{\mbox{\boldmath{$s$}}}

\newcommand{\bU}{\mbox{\boldmath{$U$}}}
\newcommand{\bu}{\mbox{\boldmath{$u$}}}
\newcommand{\bV}{\mbox{\boldmath{$V$}}}
\newcommand{\bv}{\mbox{\boldmath{$v$}}}

\newcommand{\bw}{\mbox{\boldmath{$w$}}}
\newcommand{\bX}{\mbox{\boldmath{$X$}}}
\newcommand{\bx}{\mbox{\boldmath{$x$}}}
\newcommand{\bY}{\mbox{\boldmath{$Y$}}}
\newcommand{\by}{\mbox{\boldmath{$y$}}}

\newcommand{\bSigma}{\mbox{\boldmath{$\Sigma$}}}

\newcommand{\bdelta}{\mbox{\boldmath{$\delta$}}}
\newcommand{\bDelta}{\mbox{\boldmath{$\Delta$}}}

\newcommand{\tr}{\mbox{\rm tr}\, }

\newcommand{\rank}{\mbox{\rm Rank}\, }

\newcommand{\vect}{\mbox{\rm vec}\, }

\newcommand{\ex}{{\bf\sf E}}

\def\proof {{\bf Proof:} }
\def\endproof{\hfill $\Box$}

\newtheorem{theorem}{Theorem}[section]
\newtheorem{corollary}[theorem]{Corollary}

\newtheorem{lemma}[theorem]{Lemma}
\newtheorem{proposition}[theorem]{Proposition}
\begin{document}
%\markboth{Submitted to IEEE Trans. on Signal Processing}{...}

%\title{Probabilistic Methods for the Global Optimal Solutions of QCQP}
\title{Optimal Robust Adaptive Beamforming for a General-Rank Signal Model via Equivalence of Maximin and Minimax SINR Problems}
%\footnote{
%This paper will be presented in
%part at the 2026 IEEE International Conference on Acoustics, Speech, and
%Signal Processing (ICASSP 2026), 4-8 May 2026, Barcelona, Spain. In the conference version, only Proposition III.1, a simpler version of Theorem III.10, and Corollary III.11 are presented without any detailed proofs, and no Sec. IV and the simplified Sec. II are presented. Also, only two performance metrics are compared in the simulation part. The conference paper has five pages. In the journal paper, we have enriched Secs. I and II (telling the complete story of why solving the problem is interesting), presented all proofs for Proposition III.1 to Theorem III.12, added Sec. IV, and three more performance metrics are compared in simulation part. The journal paper includes twelve pages, substantially longer than the 5-page conference paper.}

\vspace{2cm}

\author{
Yongwei Huang, {\it Senior Member, IEEE}
\thanks{This paper has been presented in
part at the 2026 IEEE International Conference on Acoustics, Speech, and
Signal Processing, 4-8 May 2026, Barcelona, Spain.}
\thanks{Y. Huang is with School of Computer Science and Engineering, Guangdong Polytechnic Normal University, Tianhe, Guangzhou, Guangdong 510665, China (email: ywhuang@gpnu.edu.cn).},
Zhenhui Huang\thanks{Z. Huang is with School of Information Engineering, Guangdong University of Technology, University Town, Guangzhou, Guangdong 510006, China (email: 2112403173@mail2.gdut.edu.cn).},
Sergiy A. Vorobyov, {\it Fellow, IEEE}\thanks{S. A. Vorobyov is with the Department of Information $\&$ Communications Engineering, School of Electrical Engineering, Aalto University, KIDE, Konemiehentie 1, 02150 Espoo, Finland (e-mail: svor@ieee.org).},
Zhi-Quan Luo, {\it Fellow, IEEE}\thanks{Z.-Q. Luo is with School of Science and Engineering, Chinese University of Hong Kong (Shenzhen), Longgang, Shenzhen, Guangdong 518172, China (email: luozq@cuhk.edu.cn).}
}

\maketitle
\begin{abstract}
The globally optimal robust adaptive beamforming (RAB) solution is studied for worst-case signal-to-interference-plus-noise ratio (SINR) maximization (the maximin SINR problem) under convex and closed uncertainty sets for the desired signal covariance and interference-plus-noise covariance (INC) matrices, considering a general-rank signal model. First, the corresponding minimax SINR problem is reformulated as a convex optimization problem. In particular, this problem becomes a semidefinite programming (SDP) problem when the uncertainty sets can be represented by finitely many linear matrix inequality constraints.
It is then shown that, for a general-rank signal model, the maximin and minimax SINR problems are equivalent when the uncertainty sets are convex and closed, in the sense that they share the same optimal value and the same set of optimal solutions. The requirement of closedness is weaker than the compactness assumption previously used to establish the equivalence between minimax and maximin SINR problems for the rank-one signal model, a state-of-the-art result reported approximately two decades ago.
Consequently, an optimal solution to the minimax SINR problem is also globally optimal for the maximin SINR problem, and this solution can be obtained by solving the equivalent SDP of the minimax problem in a single step. In contrast, existing iterative approximation algorithms for the maximin SINR problem yield only locally optimal solutions. Simulation results demonstrate that these approximation algorithms return suboptimal values that can be strictly smaller than the optimal value of the minimax problem, and that the beamformer output SINR obtained via the minimax formulation is higher than that achieved by beamformers derived from the maximin problem using approximation algorithms.
%Then, the maximin SINR and minimax SINR problems for general-rank signal model are shown equivalent to each other as the counter-parties share the equal optimal value and the same set of optimal solutions, when the uncertainty sets are convex and closed. The condition of closedness is looser than the assumption of compactness for the equivalence between the minimax and maximin SINR problems for rank-one signal model, a state-of-the-art result around two decades ago. Therefore, an optimal solution for the minimax SINR problem is globally optimal for the maximin SINR problem, and the optimal solution can be obtained by solving the equivalent SDP of the minimax problem in a single shot. In contrast, the existing iterative approximation algorithms for solving the maximin SINR problem output only a locally optimal solution. Our simulation examples demonstrate that the approximation algorithms for the maximin problem return only suboptimal values that can be strictly smaller than the optimal value of the minimax problem, and the beamformer output SINR  computed via the minimax problem is higher than that of the beamformers via the maximin problem solved by the approximation algorithms.
\end{abstract}

\vspace{0.5cm}

%\begin{center}
%\begin{small}
%{\bf Keywords.}
%\end{small}
%\end{center}

%\begin{keywords}
%Worst-case SINR maximization, general-rank signal model, equivalence between maximin and minimax SINR, convex and closed uncertainty sets, robust adaptive beamforming.
%\end{keywords}

\section{Introduction}
There has been a large number of works on robust adaptive beamforming (RAB) in the past three decades, accompanying the advances of convex optimization theory, algorithms, and solvers, as well as machine learning techniques (see \cite{icassp2026,Stoica-Li-book-beamforming, GSSBO09, Voro13survey, EMVH23, HGMW23} references therein). Considering, for example, receive beamforming, the traditional minimum variance distortionless response (MVDR) beamformer, which is simply a solution of the signal-to-interference-plus-noise ratio (SINR) maximization problem, may lead to substantially degraded performance due to uncertain parameters of the desired signal covariance for general-rank signals and the interference-plus-noise covariance (INC). This is because of a lack of perfect knowledge of the desired signal(s) and interference, distortion of the sensor array, imperfect array calibration, and others. Accordingly, RAB techniques are designed to combat the uncertainties of the parameters and improve the array performance (in terms of, e.g., the array output SINR).

Among the existing RAB approaches, the beamforming solutions for the worst-case SINR maximization (i.e., the maximin SINR) problem (for the rank-one signal covariance \cite{VGL03tsp, Lorenz-Boyd05, HFSL-23-tsp} and the general-rank signal covariance \cite{SGLW03, Amin-tsp19}) and the minimax SINR problem \cite{Kim08, Kim06-bc, Kim08-siamopt, Beck-Eldar-tsp07, HZV2019tsp} are important.\footnote{We are not able to mention other milestone RAB works developing other approaches because of space limitations, unfortunately.} The corresponding optimization problem formulations are intuitive and simple, but challenging too. In a seminal paper \cite{Kim08}, the equivalence between the minimax SINR problem and the maximin SINR problem (for a point source or far-field signal, i.e., for the rank-one signal covariance scenario) is proved under the assumption that the uncertainty sets of the two parameters (one being the actual steering vector of the desired signal and the other being the INC matrix) are convex and compact. In addition, the two problems can be turned into a semidefinite programming (SDP) problem when the uncertainty sets can be represented by finitely many linear matrix inequality (LMI) constraints. However, in most applications, not all uncertainty sets are convex and compact. For example, the steering vector Euclidean norm must equal the square root of the number of array sensors, but such an equality constraint is nonconvex \cite{HFSL-23-tsp}. In this case, both the maximin SINR and minimax SINR problems are nonconvex and hard to solve in general.

For the general-rank signal model, the equivalence between the maximin SINR and minimax SINR problems remains vague under some conditions on the uncertainty sets. The existing works on RAB designs are based on solving the worst-case SINR maximization problem, i.e., the maximin SINR problem with convex uncertainty sets \cite{SGLW03}. In \cite{SGLW03}, an estimate of the covariance matrix of the desired signal is decomposed into a tall matrix times its Hermitian, and then an error term is included in the tall matrix (and also its Hermitian). In this way, the matrix rank remains equal to or smaller than that of the estimated covariance of the desired signal. Subsequently, more and more efficient algorithms have been proposed \cite{chen-gershman08, KV13-tsp, HV18-spl}, where, in the former two papers, SDP-based iterative approximation algorithms have been developed, while the last paper has built a second-order cone programming (SOCP)-based iterative procedure, which runs faster but achieves only a locally optimal value. However, it has never been proved that the locally optimal solutions obtained by the methods proposed in the previous references are globally optimal. In \cite{Amin-tsp19}, sparse beamforming has been considered for multiple desired signals (which results in a general-rank covariance) without uncertainty considerations on the parameters.

In this paper, we finally answer the question under what general condition the maximin SINR and minimax SINR problems for the general-rank desired signal model are equivalent to each other, that is, under what condition they share the same optimal value and the same set of optimal solutions. Toward this end, we study the maximin SINR problem discussed in \cite{SGLW03, chen-gershman08, KV13-tsp, HV18-spl}, and show that if the uncertainty sets of the desired signal covariance and the INC are both convex and closed, the two maximin and minimax SINR problems are equivalent to each other. On the other hand, the minimax SINR problem can be reformulated into a convex problem without an explicit beamvector optimization variable, and then an optimal beamvector can be retrieved for the minimax SINR problem from the optimal solution of the equivalent convex problem. Such solution is then optimal for the maximin SINR problem too. To showcase the equivalence for any convex and closed uncertainty sets for the covariance matrices in the SINR formula, we present two types of uncertainty sets for the INC matrix. One is the set defined by a similarity constraint and a positive semidefinite (PSD) constraint, and the other set includes one more double-sided linear constraint on the trace of the INC (which is due to a radar application \cite{ADHP-16TSP}). Also, we discuss how to handle different problem reformulations of the minimax SINR problem when the matrix norm is the spectral norm or the Frobenius norm.

We summarize the contributions of this work as follows.
\begin{itemize}
\item We show the equivalence between the maximin SINR and minimax SINR problems for the general-rank signal model under the condition that the uncertainty sets of the desired signal covariance matrix and the INC matrix are convex and closed. Further, the optimal solution sets for the minimax and maximin SINR problems are the same, which implies that the globally optimal RAB solution for the maximin SINR (i.e., the worst-case SINR maximization problem) can be obtained by solving its counterpart minimax problem.
\item We prove that the minimax SINR problem is convex, provided that the uncertainty sets of the two covariances are convex, and it is an SDP when the uncertainty sets can be characterized by finitely many LMI constraints. Further, the RAB solution for the minimax SINR problem can be constructed from an optimal solution for the SDP in some way. Therefore, we can conclude that this RAB solution is also a globally optimal solution for the maximin SINR problem and can be obtained by solving an SDP in a single shot, rather than by using iterative approximation algorithms (see \cite{SGLW03, chen-gershman08, KV13-tsp, HV18-spl} and references therein). %studying the RAB designs via the maximin SINR criterion.
\item We can also reexpress an optimal solution for the minimax SINR problem as a saddle point of the SINR function, via the equivalence between the maximin and minimax SINR problems for the general-rank signal model under the assumption of convex and closed uncertainty sets. This result extends the earlier study about the equivalence between the two counterpart problems for the rank-one signal model \cite{Kim08, Kim06-bc, Kim08-siamopt}, and relaxes the compactness requirement to closedness for the uncertainty sets of the two parameters in the SINR formula, i.e., the desired signal and INC covariance matrices.
\item In order to show how the uncertainty set of the INC affects the performance of the RAB solution, two practical uncertainty sets of the INC are considered, and finally our simulation examples validate the theoretical results.
\end{itemize}

   The rest of the paper is organized as follows. The signal model and problem formulation for the worst-case SINR maximization (i.e., the maximin SINR problem) are stated for the general-rank signal model in Section~II. In Section~III, we obtain the globally optimal RAB solution for the maximin SINR problem via the equivalence between the maximin and minimax SINR problems. The latter problem is reformulated as a convex problem under the condition of convex and closed uncertainty sets. In Section~IV, we show how to reformulate the minimax SINR problem into an explicit SDP problem when there is one additional constraint on the energy of the interference and noise, and when the matrix norm in the similarity constraint is either the Frobenius norm or the spectral norm. Simulation examples are presented in Section~V to compare the performance of the RAB solutions obtained by methods in the literature with the proposed methods. Finally, Section~VI draws our conclusions.

{\it Notation}: We adopt the notation of using boldface for vectors
$\ba$ (lower case), and matrices $\bA$ (upper case). The transpose
operator and the conjugate transpose operator are denoted by the
superscripts $(\cdot)^T$ and $(\cdot)^H$, respectively. The notation $\mbox{tr}(\cdot)$ stands for the trace of a square matrix argument; $\bI$ and ${\bf 0}$
denote respectively the identity matrix and the matrix (or the row
vector or the column vector) with zero entries (their size is
determined from the context). The letter $j$ represents the
imaginary unit (i.e., $j=\sqrt{-1}$), while the letter $i$ often
serves as an index. For any complex number $x$, we use
$\Re(x)$ and $\Im(x)$ to denote respectively the real and
imaginary parts of $x$, $|x|$ and $\mbox{arg}(x)$ represent the
modulus and the argument of $x$, and $x^*$ ($\bx^*$ or $\bX^*$)
stands for the component-wise conjugate of $x$ ($\bx$ or $\bX$).
 %$\lambda_{\max}(\bA)$
%stands for the largest eigenvalue of $\bA$.
%The symbols $\odot$ and $\otimes$ represent
%the Hadamard element-wise and the Kronecker product,
% respectively \cite{Horn85}.
The Euclidean norm of vector $\bx$ is denoted by $\|\bx\|$, and the Frobenius norm (the spectral norm) of
matrix $\bX$ by $\|\bX\|_F$ ($\|\bX\|_2$).
%The symbol $\odot$ represents the Hadamard element-wise product.
The curled inequality symbol $\succeq$ (and its strict form $\succ$)
is used to denote generalized inequality: $\bA\succeq\bB$ meaning that
$\bA-\bB$ is a Hermitian positive semidefinite matrix
($\bA\succ\bB$ for positive definiteness).  The space of Hermitian $N\times N$ matrices
(the space of real-valued symmetric $N\times N$ matrices) is denoted by ${\cal
H}^N$ (${\cal S}^N$), the set of all positive semidefinite
matrices by ${\cal H}_+^N$ (${\cal S}_+^N$), and the set of all positive definite matrices by ${\cal H}_{++}^N$ (${\cal S}_{++}^N$).
%Assume always that $\lambda_1(\bA)\ge\lambda_2(\bA)\ge\cdots\ge\lambda_N(\bA)$ for a Hermitian matrix $\bA\in{\cal H}^N$, where $\lambda_n(\bA)$, $n=1,\ldots,n$ are the eigenvalues; namely, they are placed in a descent manner conventionally.
The notation $\ex[\cdot]$ represents the
statistical expectation. We denote by $\mathbb{R}_+^M$ the set of $M$-dimension nonnegative vectors. $\vect(\bX)$ denotes the vector stacked by the columns of $\bX$. %, and $\diag(\bx)$ denotes the diagonal matrix with the diagonal elements being the components of $\bx$.
%The notations $\range(\bX)$ and $\nul(\bX)$ stand for the range space and the null space, respectively.
The notation $\rank(\bX)$ stands for the rank of a matrix.
%Finally, $v^\star(\cdot)$  represents the optimal value of problem $(\cdot)$.
%$\bX^{1/2}$ stands for the square root matrix if $\bX\succeq\bzero$.

\section{Signal models and problem formulation}\label{Motivation-Applications}

\subsection{Signal Model}

The narrowband signal observed by a uniform linear array (ULA) with $N$ sensors is given by
\begin{equation}\label{array-observation-vector}
\by(t)=\bs(t)+\ii(t)+\bn(t),
\end{equation}
where $\bs(t)$, $\ii(t)$, and $\bn(t)$ are vectors, respectively, corresponding to the desired signal, interference, and array noise. Assume that the signal and the interference-plus-noise components are statistically independent. The beamformer output signal can be expressed as
\begin{equation}\label{output-signal}
x(t)=\bw^H\by(t),
\end{equation}
where $\bw\in\mathbb{C}^N$ is the vector of weight coefficients (termed as beamvector). Therefore, the array output SINR is written as
\begin{equation}\label{SINR-def-general-rank}
\mbox{SINR}=\frac{\bw^H\bR_s\bw}{\bw^H\bR_{i+n}\bw}
\end{equation}
where $\bR_s=\ex[\bs(t)\bs^H(t)]$ is the covariance matrix of the desired signal and $\bR_{i+n}=\ex[(\ii(t)+\bn(t))(\ii(t)+\bn(t))^H]$ is the INC %interference-plus-noise
 matrix. Note that covariance matrix $\bR_s$ is of general rank, i.e., $\rank(\bR_s)\in\{1,2,\ldots,N\}$.

In particular, when the desired signal is from a point source or a far-field source, and can be cast as $\bs(t)=s(t)\ba$ (where $s(t)$ is the desired signal waveform and $\ba$ is its steering vector) in \eqref{array-observation-vector}, the corresponding desired signal covariance matrix
$\bR_s=\ba\ba^H$
is of rank one, where, without loss of generality, the power of the desired signal is assumed to be one. Thereby, the array output SINR is rewritten as
\begin{equation}\label{SINR-def-rank-1}
\mbox{SINR}=\frac{\bw^H\ba\ba^H\bw}{\bw^H\bR_{i+n}\bw}.
\end{equation}However, in the scattered source signal case or multiple source signals case, $\bR_s$ often is no longer of rank one.

In order to find an optimal beamvector $\bw^\star$, a typical optimization problem of maximizing the SINR is adopted and can be formulated as
\begin{equation}\label{SINR-max-nonrobust}
\underset{\bw\ne\bzero}{\sf{maximize}}~\frac{\bw^H\bR_s\bw}{\bw^H\bR_{i+n}\bw},
\end{equation} as long as both the covariances $\bR_s$ and $\bR_{i+n}$ are available to the array. It is known that the optimal value for the SINR maximization problem \eqref{SINR-max-nonrobust} is the maximal eigenvalue $\lambda_{1}(\bR_{i+n}^{-\frac{1}{2}}\bR_s\bR_{i+n}^{-\frac{1}{2}})$\footnote{Throughout the paper, all eigenvalues $\lambda_n(\bX)$, $n= 1, \cdots, N$ of Hermitian matrix $\bX\in\mathbb{C}^{N\times N}$ are placed in a descending order, namely, $\lambda_1(\bX) \ge \lambda_2(\bX) \ge \cdots \ge \lambda_N(\bX)$.}, and an optimal solution for the problem is $\bw^\star=\bR_{i+n}^{-\frac{1}{2}}\bv$, where $\bv$ is an eigenvector associated with the largest eigenvalue (i.e., the principal eigenvector), provided that $\bR_{i+n}$ is positive definite.

In the rank-one signal case, the SINR maximization problem \eqref{SINR-max-nonrobust} reduces to
\begin{equation}\label{SINR-max-nonrobust-rank-1}
\underset{\bw\ne\bzero}{\sf{maximize}}~\frac{\bw^H\ba\ba^H\bw}{\bw^H\bR_{i+n}\bw}.
\end{equation}
The optimal value for problem \eqref{SINR-max-nonrobust-rank-1} is equal to $\lambda_{1}(\bR_{i+n}^{-\frac{1}{2}}\ba\ba^H\bR_{i+n}^{-\frac{1}{2}})=\ba^H\bR_{i+n}^{-1}\ba$ and the optimal solution is $\bw^\star=\bR_{i+n}^{-\frac{1}{2}}\bv=\bR_{i+n}^{-1}\ba/\|\bR_{i+n}^{-\frac{1}{2}}\ba\|$ (provided that $\ba$ is a prior information about the array), which is known as an MVDR beamvector.

However, in engineering applications, both $\bR_s$ and $\bR_{i+n}$ are always imperfectly known. In other words, to obtain an optimal beamvector for \eqref{SINR-max-nonrobust}, a presumed signal covariance matrix $\hat\bR_s$ is utilized to replace the actual signal covariance matrix $\bR_s$, and the sample covariance matrix
%\begin{equation}\label{sample-covariance}
$\hat\bR=\frac{1}{T}\sum_{t=1}^T\by(t)\by^H(t)$
%\end{equation}
is often employed to estimate $\bR_{i+n}$, as a compromise, but also more sophisticated estimates of INC matrix are widely used (see, e.g., \cite{covr-est-1,covr-est-2}). Evidently, there are certain mismatches between $\hat\bR_s$ and $\bR_s$ and between $\hat\bR$ and $\bR_{i+n}$. Therefore, applying the beamvector solution $\bw^\star=\hat\bR^{-\frac{1}{2}}\bv$ to the array, with $\bv$ being the principal eigenvector for $\hat\bR^{-\frac{1}{2}}\hat\bR_s\hat\bR^{-\frac{1}{2}}$, will lead to substantially degraded output performance of the array, which has been known for decades (see, e.g., \cite{Stoica-Li-book-beamforming}).
%is obtained by computing the SINR maximization problem
%\begin{equation}\label{SINR-max-nonrobust-presumed}
%\underset{\bw\ne\bzero}{\sf{maximize}}~\frac{\bw^H\hat\bR_s\bw}{\bw^H\hat\bR\bw},
%\end{equation}and equal to $\hat\bR^{-\frac{1}{2}}\bv$ with .

%or, for the rank-one desired signal model,
%\begin{equation}\label{SINR-max-nonrobust-presumed-rank-1}
%\underset{\bw\ne\bzero}{\sf{maximize}}~\frac{\bw^H\hat\ba\hat\ba^H\bw}{\bw^H\hat\bR\bw}.
%\end{equation}

Therefore, in order to improve beamformer performance, many RAB techniques have been proposed in the past (see, e.g., \cite{EMVH23, HGMW23}), based on advances in convex optimization techniques and algorithms (see, e.g., \cite{GSSBO09, Voro13survey, EMVH23}). Among these techniques, the worst-case SINR maximization-based RAB technique is particularly interesting and popular, either for the general-rank or the rank-one signal cases. In such RAB technique, an optimization problem of maximizing the minimal SINR is studied:
\begin{equation}\label{max-min-2008-general-rank}
\begin{array}[c]{cl}
\underset{\bw\ne\bzero}{\sf{maximize}} &\underset{\bDelta_1\in{\cal
B}_1,\bDelta_2\in{\cal B}_2}{\sf{minimize}}
\begin{array}[c]{c}\displaystyle\frac{\bw^H(\hat\bR_s+\bDelta_2)\bw}{\bw^H(\hat\bR+\bDelta_1)\bw},
\end{array}
\end{array}%\right.
\end{equation}
where the error sets ${\cal B}_1$ and ${\cal B}_2$ are predefined. For example,
%\begin{equation}\label{example-A}
%{\cal A}=\{\bdelta~|~\|\bdelta\|^2\le\epsilon,\,\|\hat\ba+\bdelta\|^2=N \}
%\end{equation}which is a nonconvex set;
\begin{equation}\label{example-B1}
{\cal B}_1=\{\bDelta_1\in\mathbb{C}^{N\times N}~|~\|\bDelta_1\|_F^2\le\gamma_1,\,\hat\bR+\bDelta_1\succeq\bzero\},
\end{equation}
and
\begin{equation}\label{example-B2}
{\cal B}_2=\{\bDelta_2\in\mathbb{C}^{N\times N}~|~\|\bDelta_2\|_F^2\le\gamma_2,\,\hat\bR_{s}+\bDelta_2\succeq\bzero\}.
\end{equation}%where $\|\cdot\|_F$ represents the Frobenius norm.
However, finding a globally optimal solution $\bw^\star$ for the maximin problem \eqref{max-min-2008-general-rank}, even when the sets ${\cal B}_1$ and ${\cal B}_2$ are convex, remains an open problem. In particular, a simple and practical sufficient condition for achieving global optimality has yet to be established (cf. \cite{Verdu-1984, HVL-20-tsp}).

Alternatively, the following worst-case SINR maximization problem for the general-rank signal case is considered:
\begin{equation}\label{max-min-2018-general-rank-0}
\begin{array}[c]{cl}
\underset{\bw\ne\bzero}{\sf{maximize}} &\underset{\bDelta_0\in{\cal
B}_0,\bDelta_1\in{\cal B}_1}{\sf{minimize}}
\begin{array}[c]{c}\displaystyle\frac{\bw^H(\hat\bQ+\bDelta_0)(\hat\bQ+\bDelta_0)^H\bw}{\bw^H(\hat\bR+\bDelta_1)\bw},
\end{array}
\end{array}%\right.
\end{equation}where $\hat\bR_s=\hat\bQ\hat\bQ^H$, $\hat\bQ\in\mathbb{C}^{N\times M}$, $N\ge M=\rank(\hat\bR_s)$, and ${\cal B}_0$ is a given error set. For instance,
\begin{equation}\label{example-B0}
{\cal B}_0=\{\bDelta_0\in\mathbb{C}^{N\times M}~|~\|\bDelta_0\|_F^2\le\epsilon\},
\end{equation}or, more generally,
\begin{equation}\label{example-B0-general}
{\cal B}_0=\{\bDelta_0\in\mathbb{C}^{N\times M}~|~\tr(\bDelta_0\bA_i\bDelta_0^H)\le\epsilon_i,\,i=1,\ldots,I\},
\end{equation}where $\bA_i\in\mathbb{C}^{M\times M}$, $i=1,\ldots,I$, are Hermitian matrices. However, the error set for $\bDelta_1$ remains unchanged. Therefore, one of the advantages of studying the problem formulation in \eqref{max-min-2018-general-rank-0} is that the rank of the signal covariance matrix is guaranteed to be no greater than that of $\hat{\bR}_s$.

%{\color{red}Consider the case of $\{\bQ\bQ^H~|~\bQ=\hat\bQ+\bDelta,\,\bDelta=[\bdelta_1,\ldots,\bdelta_M],\,\|\bdelta_m\|^2\le\zeta_m,\,m=1,\ldots,M\}$ (Amin paper 2019 tsp).}

In \cite{KV13-tsp}, the authors show that a globally optimal solution for problem \eqref{max-min-2018-general-rank-0}, with the special set ${\cal B}_0$ defined in \eqref{example-B0}, can be obtained provided that the radius $\sqrt{\epsilon}$ is sufficiently small. In this paper, we extend their result and establish a simpler global optimality condition for problem \eqref{max-min-2018-general-rank-0}, while also considering more general settings.

In particular, when $\hat{\bR}_s = \hat{\ba} \hat{\ba}^H$ is a rank-one matrix (corresponding to a point source or a far-field signal scenario, where $\hat{\ba}$ is the presumed steering vector of the desired signal), problem \eqref{max-min-2018-general-rank-0} reduces to the problem for the rank-one signal model
\begin{equation}\label{max-min-2008-rank-one}
\begin{array}[c]{cl}
\underset{\bw\ne\bzero}{\sf{maximize}} &\underset{\bdelta\in{\cal
A},\bDelta_1\in{\cal B}_1}{\sf{minimize}}
\begin{array}[c]{c}\displaystyle\frac{\bw^H(\hat\ba+\bdelta)(\hat\ba+\bdelta)^H\bw}{\bw^H(\hat\bR+\bDelta_1)\bw},
\end{array}
\end{array}%\right.
\end{equation}where the error set ${\cal A}$ is either convex or nonconvex. A practical and nonconvex example is
\begin{equation}\label{example-A}
{\cal A}=\{\bdelta\in\mathbb{C}^{N}~|~\|\bdelta\|^2\le\epsilon,\,\|\hat\ba+\bdelta\|^2=N \}.
\end{equation}%which is a nonconvex set;

It is noted that the results claimed in \cite{Kim08,Kim06-bc,Kim08-siamopt}\footnote{It is shown therein that problem \eqref{max-min-2008-rank-one} is a convex optimization problem when the uncertainty set ${\cal A}\times{\cal B}_1$ is compact and convex. In particular, \eqref{max-min-2008-rank-one} can be reformulated as an SDP when ${\cal A}$ and ${\cal B}_1$ can be represented by finitely many LMIs.} are not applicable to problem \eqref{max-min-2008-rank-one}, since the set ${\cal A}$ defined in \eqref{example-A} is nonconvex. % (but \cite{HFSL-23-tsp} discussed efficient algorithms to approximately solve it).
%has been studied in \cite{Kim08,Kim06-bc,Kim08-siamopt}. In the paper, we

To proceed and simplify notations, let
\begin{equation}\label{define-R1-Q}
\bQ=\hat\bQ+\bDelta_0,\quad\bR_1=\hat\bR+\bDelta_1,
\end{equation}
in problem \eqref{max-min-2018-general-rank-0}. %and let
%\begin{equation}\label{define-a}
%\ba=\hat\ba+\bdelta
%\end{equation}
%in \eqref{max-min-2008-rank-one}.
Thereby, the maximin problem \eqref{max-min-2018-general-rank-0} is rewritten as
\begin{equation}\label{max-min-2018-general-rank-1}
\begin{array}[c]{cl}
\underset{\bw\ne\bzero}{\sf{maximize}} &\underset{\bQ\in\bar{\cal
B}_0,\bR_1\in\bar{\cal B}_1}{\sf{minimize}}
\begin{array}[c]{c}\displaystyle\frac{\bw^H\bQ\bQ^H\bw}{\bw^H\bR_{1}\bw}.
\end{array}
\end{array}%\right.
\end{equation}
%and
%\begin{equation}\label{max-min-2008-rank-one-1}
%\begin{array}[c]{cl}
%\underset{\bw\ne\bzero}{\sf{maximize}} &\underset{\ba\in\bar{\cal
%A},\bR_1\in\bar{\cal B}_1}{\sf{minimize}}
%\begin{array}[c]{c}\displaystyle\frac{\bw^H\ba\ba^H\bw}{\bw^H\bR_{1}\bw}.
%\end{array}
%\end{array}%\right.
%\end{equation}
The uncertainty set $\bar{\cal B}_0$, for instance, can be specified as
\begin{equation}\label{example-B0-bar}
\tilde{\cal B}_0=\{\bQ\in\mathbb{C}^{N\times M}~|~\tr((\bQ-\hat\bQ)\bA_i(\bQ-\hat\bQ)^H)\le\epsilon_i,\,i=1,\ldots,I\}
\end{equation}(transformed from \eqref{example-B0-general}). Similarly, the uncertainty set $\bar{\cal B}_1$ can be detailed as
\begin{equation}\label{example-B1-bar}
\tilde{\cal B}_1=\{\bR_1\in{\cal H}^{N}~|~\|\bR_1-\hat\bR\|_F^2\le\gamma_1,\,\bR_{1}\succeq\bzero \}
\end{equation}(reformulated from \eqref{example-B1}).

It is known (see, e.g., \cite{HV18-spl}) that problem \eqref{max-min-2018-general-rank-1} with the special $\tilde{\cal B}_0$ of \eqref{example-B0-bar}:
\begin{equation}\label{example-B0-bar-1}
\tilde{\cal B}_0=\{\bQ\in\mathbb{C}^{N\times M}~|~\|\bQ-\hat\bQ\|_F^2\le\eta\},
\end{equation}
and $\tilde{\cal B}_1$ defined in \eqref{example-B1-bar} can be solved approximately by a sequence of SOCPs.

Note that if $\|\hat{\bQ}\|_F^2 \le \eta$ (which implies that $\bQ = \mathbf{0}$ is feasible), the optimal value of problem \eqref{max-min-2018-general-rank-1} is zero, meaning that the array output SINR is also zero. However, this is not a practical scenario. Therefore, we exclude the case of $\bQ = \mathbf{0}$ by assuming that $\eta < \|\hat{\bQ}\|_F^2$.
Then problem \eqref{max-min-2018-general-rank-1} is equivalently transformed into
\begin{equation}\label{max-min-2018-general-rank-2}
\begin{array}[c]{cl}
\underset{\bw\ne\bzero}{\sf{maximize}} &\underset{\bQ\in\tilde{\cal
B}_0}{\sf{minimize}}
\begin{array}[c]{c} \bw^H\bQ\bQ^H\bw
\end{array}\\
\sf{subject\;to}               &
 \bw^H(\hat\bR+\sqrt{\gamma_1}\bI)\bw\le1.
\end{array}%\right.
\end{equation}

Since the optimal value of the inner minimization problem is equal to
\begin{equation}\label{inner-min-optimal-v-1}
\underset{\bQ\in\tilde{\cal B}_0}{\sf{minimize}}~~ \|\bQ^H\bw\|
 = \max \{\|\hat\bQ^H\bw\|-\sqrt{\eta}\|\bw\|,0\}
\end{equation}
(see, e.g., \cite{Beck09-orl,HV22-sp}), hence problem \eqref{max-min-2018-general-rank-2} can be further reexpressed as the following nonconvex problem:
\begin{equation}\label{max-min-2018-general-rank-3}
\begin{array}[c]{cl}
\underset{\bw\ne\bzero}{\sf{maximize}} & \|\hat\bQ^H\bw\|-\sqrt{\eta}\|\bw\|\\
\sf{subject\;to}               &
 \bw^H\hat\bR\bw+\sqrt{\gamma_1}\|\bw\|^2\le1.
\end{array}%\right.
\end{equation}
This problem has been solved by iteratively solving a sequence of SDPs \cite{KV13-tsp} or SOCPs \cite{HV18-spl}. It is noted that, in the numerical results, the iterative procedure of either SDPs in \cite{KV13-tsp} or SOCPs in \cite{HV18-spl} achieves global optimality for problem \eqref{max-min-2018-general-rank-1}, even though only local optimality is guaranteed theoretically.

Note that if the Frobenius norm in the uncertainty set \eqref{example-B0-bar-1} is changed to the spectral norm, namely the following set is considered
\begin{equation}\label{example-B0-bar-1-matrx2norm}
\hat{\cal B}_0=\{\bQ\in\mathbb{C}^{N\times M}~|~\|\bQ-\hat\bQ\|_2^2\le\eta\},
\end{equation}
the optimal value in \eqref{inner-min-optimal-v-1} of the inner minimization problem in \eqref{max-min-2018-general-rank-2} remains unaltered. Also, if the Frobenius norm in the uncertainty set \eqref{example-B1-bar} is changed to the spectral norm, namely,
\begin{equation}\label{example-B1-hat}
\hat{\cal B}_1=\{\bR_1\in{\cal H}^{N}~|~\|\bR_1-\hat\bR\|_2^2\le\gamma_2,\,\bR_{1}\succeq\bzero \},
\end{equation}
then it still follows that
\begin{equation}\label{max-IN-power}
\bw^H\hat\bR\bw+\sqrt{\gamma_2}\|\bw\|^2=\underset{\bR_1\in\hat{\cal B}_1}{\sf{maximize}}~\bw^H\bR_{1}\bw.
\end{equation}
Accordingly, problem \eqref{max-min-2018-general-rank-1} with uncertainty sets $\hat{\cal B}_0$ and $\hat{\cal B}_1$ defined in \eqref{example-B0-bar-1-matrx2norm} and \eqref{example-B1-hat}, respectively, or with sets $\tilde{\cal B}_0$ and $\tilde{\cal B}_1$ defined in \eqref{example-B0-bar-1} and \eqref{example-B1-bar}, respectively, can be  recast into nonconvex problem \eqref{max-min-2018-general-rank-3}. In other words, the equivalence between the maximin problem \eqref{max-min-2018-general-rank-1} and problem \eqref{max-min-2018-general-rank-3} remains valid if all matrix norms are either Frobenius norms or spectral norms.

%In our previous works \cite{KV13-tsp,HV18-spl}, we have studied to how to solve nonconvex problem \eqref{max-min-2018-general-rank-3} up to a local optimality.
In this paper, we address the problem of efficiently obtaining globally optimal solutions for the maximin problem \eqref{max-min-2018-general-rank-1} in a more general setting, where the sets $\bar{\cal B}_0$ and $\bar{\cal B}_1$ are convex and closed. Our approach generalizes previous works \cite{KV13-tsp, HV18-spl}, as well as the studies in \cite{Kim08, Kim06-bc, Kim08-siamopt}, which focus on the rank-one signal model in problem \eqref{max-min-2008-rank-one} under the assumption that ${\cal A}$ and ${\cal B}_1$ are convex and compact.

\section{Globally Optimal Solution for the Maximin SINR Problem via Its Minimax SINR Problem for the General-Rank Signal Model}

In this section, we study a convex equivalent of the worst-case SINR maximization problem \eqref{max-min-2018-general-rank-1}. To proceed, we assume that both $\bar{\cal B}_0 \subset \mathbb{C}^{N \times M} \setminus \{\mathbf{0}\}$ and $\bar{\cal B}_1 \subset {\cal H}_{++}^N$ in \eqref{max-min-2018-general-rank-1} are convex and closed, without restricting them to the specific forms in \eqref{example-B0-bar} and \eqref{example-B1-bar}, respectively.

\subsection{Upper Bound for the Worst-Case SINR Maximization Problem}
The following minimax problem provides an upper bound on the optimal value of the maximin problem \eqref{max-min-2018-general-rank-1}
\begin{equation}\label{min-max-2018-general-rank-1}
\begin{array}[c]{cl}
\underset{\bQ\in\bar{\cal
B}_0,\bR_1\in\bar{\cal B}_1}{\sf{minimize}} & \underset{\bw\ne\bzero}{\sf{maximize}}
\begin{array}[c]{c}\displaystyle\frac{\bw^H\bQ\bQ^H\bw}{\bw^H\bR_{1}\bw}.
\end{array}
\end{array}%\right.
\end{equation}

%Before we claim that the minimax problem \eqref{min-max-2018-general-rank-1} is a convex optimization problem, it is assumed that $\lambda_1(\bA)\ge\lambda_2(\bA)\ge\cdots\ge\lambda_N(\bA)$ for a Hermitian matrix $\bA\in{\cal H}^N$; in other words, the eigenvalues are placed in a descent manner for each Hermitian matrix.

\begin{proposition}\label{min-max-convex-opt-problem}
Suppose that $\bar{\cal B}_0\subset\mathbb{C}^{N\times M}/\{\bzero\}$ and $\bar{\cal B}_1\subset{\cal H}_{++}^N$ are convex and closed. Then the minimax problem \eqref{min-max-2018-general-rank-1} is equivalent to the following convex optimization problem
\begin{equation}\label{min-max-2018-general-rank-1-inner-max-dual}
\begin{array}[c]{cl}
\underset{\lambda,\bQ,\bR_1}{\sf{minimize}} & \lambda\\
\sf{subject\;to}           & \left[\begin{array}{cc}\bR_{1}&\bQ\\ \bQ^H&\lambda\bI\end{array}\right]\succeq\bzero\\
                           & \bQ\in\bar{\cal B}_0,\,\bR_1\in\bar{\cal B}_1\\
                           &\lambda\in{\mathbb{R}},\,\bQ\in\mathbb{C}^{N\times M},\,{\color{black}\bR_1}\in{\cal H}^N,
\end{array}%\right.
\end{equation}%where ${\cal H}^N$ is the set of all Hermitian matrices of $N\times N$. The problem
which is solvable\footnote{By saying ``solvable", we mean that the minimization problem is feasible, bounded below and the optimal value is attained at a feasible solution (cf. \cite[page 2]{Nemi-book-2001}).}. Suppose that $(\lambda^\star,\bQ^\star,\bR_1^\star)$ is an optimal solution for \eqref{min-max-2018-general-rank-1-inner-max-dual}. Then, $(\bw^\star, \bQ^\star, \bR_1^\star)$ is optimal for the minimax problem \eqref{min-max-2018-general-rank-1} with $\bw^\star=\bR_1^{\star -\frac{1}{2}}\bu_1^\star$, where $\bu_1^\star$ is an eigenvector associated with $\lambda_1(\bR_1^{\star -\frac{1}{2}}\bQ^\star\bQ^{\star H}\bR_1^{\star -\frac{1}{2}})=\lambda^\star$, which is the optimal value of problem \eqref{min-max-2018-general-rank-1-inner-max-dual}.
\end{proposition}

See Appendix~\ref{proof-prop-min-max-convex-opt-problem} for the proof.

It is worth noting that if $\mathbf{0} \in \bar{\cal B}_0$, then problem \eqref{min-max-2018-general-rank-1-inner-max-dual} admits a trivial solution: $\lambda^\star = 0$, $\mathbf{Q}^\star = \mathbf{0}$, and any $\mathbf{R}^\star \in \bar{\cal B}_1$. To avoid this trivial case, we therefore assume that $\mathbf{0} \notin \bar{\cal B}_0$. % (but, if we do not assume it, then \eqref{min-max-2018-general-rank-1-inner-max-dual} includes always the trivial case, which appears okay).

Clearly, if $\bar{\cal B}_0$ and $\bar{\cal B}_1$ can be represented by finitely many LMIs, then problem \eqref{min-max-2018-general-rank-1-inner-max-dual} is an SDP. In other words, the minimax problem \eqref{min-max-2018-general-rank-1} is equivalent to an SDP. For example, $\bar{\cal B}_0$ can be defined as in \eqref{example-B0-bar} with each $\mathbf{A}_i$ being PSD, or as in \eqref{example-B0-bar-1}, and $\bar{\cal B}_1$ can be defined as in \eqref{example-B1-bar} or as
\begin{equation}\label{example-B1-tilde-trace-l-u}
\bar{\cal B}_1^\prime=\{\bR_1\in{\cal H}^{N}~|~\|\bR_1-\hat\bR\|_F^2\le\gamma_1,\,\rho_1\le\tr\bR_1\le\rho_2,\,\bR_{1}\succeq\bzero \},
\end{equation}
the parameters $\rho_1$ and $\rho_2$ determine the trust interval of the interference and noise energy (cf. \cite{ADHP-16TSP}). We will return to consider the uncertainty set in \eqref{example-B1-tilde-trace-l-u} later, but here we note that all of these uncertainty sets, $\bar{\cal B}_0$ and $\bar{\cal B}_1$, can be represented by LMIs.

%In particular, if $\bQ\bQ^H$ is of rank one, namely, $\bQ=\ba$ is the vector as in \eqref{max-min-2008-rank-one-1}, then from Proposition \ref{min-max-convex-opt-problem} it follows that $\bw^\star=\bR_1^{\star -\frac{1}{2}}\bu_1^\star$ (where $\bu_1^\star=\bR_1^{\star -\frac{1}{2}}\ba^\star/\|\bR_1^{\star -\frac{1}{2}}\ba^\star\|$) is equal to $\alpha^\star\bR_1^{\star -1}\ba^\star$ with $\alpha^\star=1/\|\bR_1^{\star -\frac{1}{2}}\ba^\star\|$, which is the well-known MVDR beamformer for the rank-one signal model.

\subsection{Equivalence between the Worst-Case SINR Maximization Problem and Its Minimax Counterpart Problem}
Let $v_1^\star$ denote the optimal value of the worst-case SINR maximization problem \eqref{max-min-2018-general-rank-1}, and let $v_2^\star$ denote the optimal value of its minimax counterpart problem \eqref{min-max-2018-general-rank-1}. In this subsection, we aim to identify conditions under which
\begin{equation}\label{opt-vals-equal}
v_1^\star\ge v_2^\star.
\end{equation}
Thus, we have $v_1^\star = v_2^\star$, since the inequality $v_2^\star \ge v_1^\star$ is trivial. We will also show that the solution constructed in Proposition~\ref{min-max-convex-opt-problem} for the minimax problem \eqref{min-max-2018-general-rank-1} is optimal for the worst-case SINR maximization problem \eqref{max-min-2018-general-rank-1}. In other words, the globally optimal solution for the maximin (nonconvex) problem \eqref{max-min-2018-general-rank-1} can be obtained by solving the convex problem \eqref{min-max-2018-general-rank-1-inner-max-dual} in a single step, without relying on the iterative methods in \cite{KV13-tsp, HV18-spl}. % namely, the global optimality for the maximin  problem \eqref{max-min-2018-general-rank-1} can be achieved within polynomial-time complexity.

Toward this end, suppose that $(\lambda^\star, \bQ^\star, \bR_1^\star)$ is optimal for \eqref{min-max-2018-general-rank-1-inner-max-dual}. It follows from Proposition~\ref{min-max-convex-opt-problem} that $(\bw^\star, \bQ^\star, \bR_1^\star)$ is an optimal solution for problem \eqref{min-max-2018-general-rank-1}, with
\begin{equation}\label{define-w-star}
\bw^\star=\bR_1^{\star -\frac{1}{2}}\bu_1^\star,
\end{equation}
where $\bu_1^\star$ is an eigenvector associated with the largest eigenvalue $\lambda_1(\bR_1^{\star -\frac{1}{2}}\bQ^\star\bQ^{\star H}\bR_1^{\star -\frac{1}{2}})$. Note that
\eqref{min-max-2018-general-rank-1-inner-max-dual} can be rewritten as
\begin{equation}\label{min-max-2018-general-rank-1-inner-max-dual-equv-1}
\begin{array}[c]{cl}
\underset{\bQ,\bR_1}{\sf{minimize}} & \lambda_1(\bR_1^{-\frac{1}{2}}\bQ\bQ^{H}\bR_1^{-\frac{1}{2}})\\
\sf{subject\;to}           & \bQ\in\bar{\cal B}_0,\,\bR_1\in\bar{\cal B}_1.
\end{array}%\right.
\end{equation}
Thus, it is straightforward to verify that $(\mathbf{Q}^\star, \mathbf{R}_1^\star)$ is an optimal solution of problem \eqref{min-max-2018-general-rank-1-inner-max-dual-equv-1}, with the optimal value $\lambda^\star=\lambda_1(\bR_1^{\star -\frac{1}{2}}\bQ^\star\bQ^{\star H}\bR_1^{\star -\frac{1}{2}})$, which leads to the following chain of equalities:
\begin{eqnarray}\label{the-optimal-values-0}
v_2^\star=\lambda^\star&=&\lambda_1(\bR_1^{\star -\frac{1}{2}}\bQ^\star\bQ^{\star H}\bR_1^{\star -\frac{1}{2}})\\ \label{max-eigv-expression}
                       &=&\bu_1^{\star H}\bR_1^{\star -\frac{1}{2}}\bQ^\star\bQ^{\star H}\bR_1^{\star -\frac{1}{2}}\bu_1^\star\\ \label{w-Q-QH-w-all-star}
                       &=&\bw^{\star H}\bQ^\star\bQ^{\star H}\bw^\star\\ \label{the-optimal-values-1}
                       &=&\frac{\bw^{\star H}\bQ^\star\bQ^{\star H}\bw^\star}{\bw^{\star H}\bR^{\star}_1\bw^\star},
\end{eqnarray}where the last equality is thanks to $\bw^{\star H}\bR^{\star}_1\bw^\star=\|\bu_1^\star\|^2=1$.

Therefore, in order to show \eqref{opt-vals-equal}, it suffices to prove that
\begin{equation}\label{opt-vals-equal-1}
v_1^\star\ge\bw^{\star H}\bQ^\star\bQ^{\star H}\bw^\star.
\end{equation}
It is known (see, e.g., \cite[Example 18.c, page 147]{Nemi-book-2001}) that $\lambda_1(\bX)$ is a convex function in ${\cal H}^N$. Let
\begin{equation}\label{the-objective-1}
f(\bQ,\bR_1)=\lambda_1(\bR_1^{-\frac{1}{2}}\bQ\bQ^{H}\bR_1^{-\frac{1}{2}}).
\end{equation}
We further claim that $f$ is also a convex function.

\begin{lemma}\label{lambda-1-cvx-fcn} The objective function $f(\bQ, \bR_1)$ defined in \eqref{the-objective-1} is convex for  $(\bQ, \bR_1) \in \mathbb{C}^{N\times M} \times {\cal H}_{++}^N$.
\end{lemma}

See Appendix~\ref{proof-lemma-lambda-1-cvx-fcn} for the proof.

Since problem \eqref{min-max-2018-general-rank-1-inner-max-dual-equv-1} is a standard convex optimization problem, any stationary point of this problem is also globally optimal (see, e.g., \cite[Theorem 9.7]{Beck-book-2014}). In other words, the following condition characterizes the optimality for problem \eqref{min-max-2018-general-rank-1-inner-max-dual-equv-1}.

\begin{proposition}\label{optimality-condition-1}
The solution $(\bQ^\star, \bR_1^\star) \in \bar{\cal B}_0 \times \bar{\cal B}_1$, where $\bar{\cal B}_0\subset\mathbb{C}^{N\times M}$ and $\bar{\cal B}_1\subset{\cal H}_{++}^N$ are convex and closed, is optimal for \eqref{min-max-2018-general-rank-1-inner-max-dual-equv-1} if and only if the following condition
\begin{equation}\label{optimality-conditions-1}
\begin{array}{l}
\Re\left(\tr\left(\left(\frac{\partial f}{\partial \bQ}\left|_{(\bQ^\star,\bR_1^\star)}\right.\right)^H(\bQ-\bQ^\star)\right)\right)+\\
~~~~~~~~~~~~~~~~~~~~~~~~~~\Re\left(\tr\left(\left(\frac{\partial f}{\partial \bR_1}\left|_{(\bQ^\star,\bR_1^\star)}\right.\right)^H(\bR_1-\bR_1^\star)\right)\right)\\
~~~~~~~~~~~~~~~~~~~~~~~~~~~~~~~~~~~~~~~~~~~~~~~\ge0,\,\forall\bQ\in\bar{\cal B}_0,\,\bR_1\in\bar{\cal B}_1,
\end{array}
\end{equation}
holds.
\end{proposition}

In order to simplify the optimality condition \eqref{optimality-conditions-1}, we compute the two partial derivatives:
\begin{equation}\label{two-partial-derivatives}
\frac{\partial f}{\partial \bQ}\left|_{(\bQ^\star,\bR_1^\star)}\right. \; \mbox{and } \; \frac{\partial f}{\partial \bR_1}\left|_{(\bQ^\star,\bR_1^\star)}\right..
\end{equation}
Toward this end, we first present the following proposition.

\begin{proposition}\label{first-partial-derivative-1}
It holds that
\begin{equation}\label{first-partial-derivative-val-1}
\frac{\partial f}{\partial \bQ}\left|_{(\bQ^\star,\bR_1^\star)}\right.=2\bw^\star\bw^{\star H}\bQ^\star,
\end{equation}
where $\bw^\star=\bR_1^{\star -\frac{1}{2}}\bu_1^\star$, and $\bu_1^\star$ is an eigenvector associated with $\lambda_1(\bR_1^{\star -\frac{1}{2}}\bQ^\star\bQ^{\star H}\bR_1^{\star -\frac{1}{2}})$. Moreover,
\begin{equation}\label{first-partial-derivative-val-11}
\tr\left(\left(\frac{\partial f}{\partial \bQ}\left|_{(\bQ^\star,\bR_1^\star)}\right.\right)^H(\bQ-\bQ^\star)\right)=2\bw^{\star H}(\bQ-\bQ^\star)\bQ^{\star H}\bw^\star.
\end{equation}
\end{proposition}

See Appendix~\ref{proof-prop-first-partial-derivative-1} for the proof.

To compute the second partial derivative in \eqref{two-partial-derivatives}, we make the following useful observation.

\begin{lemma}\label{svd-u1-v1-lemma}
Suppose that $\bA\in\mathbb{C}^{N\times M}$ has the singular value decomposition $\bA=\bU\bSigma\bV^H$, and let $\bu_1$ and $\bv_1$ denote the first columns of $\bU$ and $\bV$, respectively, and $\lambda_1$ stand for the maximal eigenvalue $\lambda_1(\bA\bA^H)$. Then, it holds that
\begin{equation}\label{svd-u1-v1-1}
\bA\bv_1\bv_1^H\bA^H=\lambda_1\bu_1\bu_1^H,
\end{equation}
and
\begin{equation}\label{svd-u1-v1-2}
\bA^H\bu_1\bu_1^H\bA=\lambda_1\bv_1\bv_1^H,
\end{equation}
\end{lemma}

See Appendix~\ref{proof-lemma-svd-u1-v1-lemma} for the proof.

With Lemma~\ref{svd-u1-v1-lemma} in hand, we compute $\frac{\partial f}{\partial \bR_1}$ at $(\bQ^\star,\bR_1^\star)$ as follows.

\begin{proposition}\label{second-partial-derivative-1}
It holds that
\begin{equation}\label{second-partial-derivative-val-1}
\frac{\partial f}{\partial \bR_1}\left|_{(\bQ^\star,\bR_1^\star)}\right.=-\lambda^\star\bw^\star\bw^{\star H},
\end{equation}
where $\bw^\star=\bR_1^{\star -\frac{1}{2}}\bu_1^\star$ and $\bu_1^\star$ is an eigenvector associated with $\lambda^\star=\lambda_1(\bR_1^{\star -\frac{1}{2}}\bQ^\star\bQ^{\star H}\bR_1^{\star -\frac{1}{2}})$. Moreover,
\begin{equation}\label{second-partial-derivative-val-11}
\tr\left(\left(\frac{\partial f}{\partial \bR_1}\left|_{(\bQ^\star,\bR_1^\star)}\right.\right)^H(\bR_1-\bR_1^\star)\right)=-\lambda_1^\star\bw^{\star H}(\bR_1-\bR_1^\star)\bw^\star.
\end{equation}
\end{proposition}

See Appendix \ref{proof-prop-second-partial-derivative-1} for the proof.

By combining Propositions~\ref{first-partial-derivative-1} and \ref{second-partial-derivative-1}, we arrive at the following proposition, which provides a simplified reformulation of the optimality condition \eqref{optimality-conditions-1} for problem \eqref{min-max-2018-general-rank-1-inner-max-dual-equv-1}.

\begin{proposition}\label{optimality-conditions-thm}
Suppose that $\bar{\cal B}_0\subset\mathbb{C}^{N\times M}\setminus\{\bzero\}$ and $\bar{\cal B}_1\subset{\cal H}_{++}^N$ in \eqref{min-max-2018-general-rank-1-inner-max-dual-equv-1} are convex and closed. Then, $(\bQ^\star,\bR_1^\star)$ is an optimal solution for \eqref{min-max-2018-general-rank-1-inner-max-dual-equv-1}, if and only if the condition
\begin{equation}\label{optimality-conditions-2}
\begin{array}{l}
2\Re(\bw^{\star H}\bQ\bQ^{\star H}\bw^\star)-\bw^{\star H}(\lambda^\star(\bR_1-\bR_1^\star)+2\bQ^\star\bQ^{\star H})\bw^\star\\
~~~~~~~~~~~~~~~~~~~~~~~~~~~~~~~~~~~~~~~~~~~~~~~\ge0,\,\forall\bQ\in\bar{\cal B}_0,\,\bR_1\in\bar{\cal B}_1,
\end{array}
\end{equation}
holds. Here $\bw^\star=\bR_1^{\star -\frac{1}{2}}\bu_1^\star$ and $\bu_1^\star$ is an eigenvector associated with $\lambda^\star=\lambda_1(\bR_1^{\star -\frac{1}{2}}\bQ^\star\bQ^{\star H}\bR_1^{\star -\frac{1}{2}})$.
\end{proposition}

With the previous results, we are now able to show \eqref{opt-vals-equal-1}. In fact, observe that
\begin{equation}\label{max-min-opt-val-lower-bound}
\begin{array}{l}
\underset{\bw\ne\bzero}{\sf{maximize}}~~\underset{\bQ\in\bar{\cal
B}_0,\bR_1\in\bar{\cal B}_1}{\sf{minimize}}~~\frac{\bw^H\bQ\bQ^H\bw}{\bw^H\bR_{1}\bw}\\
~~~~~~~~~~~~~~~~~~~~~~~~~~~~~~~~~~\ge\underset{\bQ\in\bar{\cal
B}_0,\bR_1\in\bar{\cal B}_1}{\sf{minimize}}~~\frac{\bw^{\star H}\bQ\bQ^H\bw^\star}{\bw^{\star H}\bR_{1}\bw^\star},
\end{array}
\end{equation}
where $\bw^\star$ is defined in \eqref{define-w-star}. In other words,
\begin{equation}\label{max-min-opt-val-lower-bound-1}
v_1^\star\ge\underset{\bQ\in\bar{\cal
B}_0,\bR_1\in\bar{\cal B}_1}{\sf{minimize}}~~\frac{\bw^{\star H}\bQ\bQ^H\bw^\star}{\bw^{\star H}\bR_{1}\bw^\star}.
\end{equation}
On the other hand, note that $v_2^\star=\bw^{\star H}\bQ^\star\bQ^{\star H}\bw^\star$ (see \eqref{the-optimal-values-0}-\eqref{w-Q-QH-w-all-star}) with $\bw^{\star H}\bR^{\star}_1\bw^\star=1$, where $\bQ^\star$ together with $\bR_1^\star$ is optimal for \eqref{min-max-2018-general-rank-1-inner-max-dual-equv-1}. Therefore, if we prove that
\begin{equation}\label{max-min-opt-val-lower-bound-2}
%\bw^{\star H}\bQ^\star\bQ^{\star H}\bw^\star=
\bw^{\star H}\bQ^\star\bQ^{\star H}\bw^\star=\underset{\bQ\in\bar{\cal
B}_0,\bR_1\in\bar{\cal B}_1}{\sf{minimize}}~~\frac{\bw^{\star H}\bQ\bQ^H\bw^\star}{\bw^{\star H}\bR_{1}\bw^\star},
\end{equation}
then it follows from \eqref{max-min-opt-val-lower-bound-1} that $v_1^\star\ge v_2^\star$, which implies $v_1^\star= v_2^\star$ taking into account \eqref{opt-vals-equal}.

Toward this end, let us  show \eqref{max-min-opt-val-lower-bound-2}. Define
\begin{equation}\label{cvx-min-problem-obj-fcn}
h(\bQ,\bR_1)=\frac{\bw^H\bQ\bQ^H\bw}{\bw^{H}\bR_{1}\bw}=\frac{\|\bQ^H\bw\|^2}{\bw^{H}\bR_{1}\bw}.
\end{equation}
for a given $\bw\ne\bzero$.
We claim that the function $h(\bQ,\bR_1)$ is convex, as stated in the following proposition.

\begin{proposition}\label{cvx-obj-fcn-prop}
For $\bw\ne\bzero$, the function defined in \eqref{cvx-min-problem-obj-fcn}
%\begin{equation}\label{cvx-min-problem-obj-fcn-general}
%h(\bQ,\bR_1)=\frac{\|\bQ^H\bw\|^2}{\bw^{H}\bR_{1}\bw}.
%\end{equation}
%Then, $h(\bQ,\bR_1)$
is convex over $\mathbb{C}^{N\times M}\times{\cal H}_+^N$.
\end{proposition}

See Appendix~\ref{proof-prop-cvx-obj-fcn-prop} for the proof.

Thus, minimization problem %\eqref{max-min-opt-val-lower-bound-2}
\begin{equation}\label{cvx-problem-1}
\underset{\bQ\in\bar{\cal
B}_0,\bR_1\in\bar{\cal B}_1}{\sf{minimize}}~~\frac{\bw^{\star H}\bQ\bQ^H\bw^\star}{\bw^{\star H}\bR_{1}\bw^\star}
\end{equation}
is convex. Moreover, the optimal solution $(\bQ^\star, \bR_1^\star)$ for \eqref{min-max-2018-general-rank-1-inner-max-dual-equv-1} is also optimal for \eqref{cvx-problem-1}. Specifically, we have the following proposition.

\begin{proposition}\label{cvx-problem-1-optsoln}
Suppose that $(\bQ^\star, \bR_1^\star)$ is optimal for \eqref{min-max-2018-general-rank-1-inner-max-dual-equv-1}, and let $\bw^\star$ be defined as in \eqref{define-w-star}. Then, $(\bQ^\star, \bR_1^\star)$ is also optimal for \eqref{cvx-problem-1} with the optimal value $\bw^{\star H} \bQ^\star \bQ^{\star H} \bw^\star$, and vice versa.
\end{proposition}

See Appendix~\ref{proof-prop-cvx-problem-1-optsoln} for the proof.

\subsection{Saddle Point Reformulation for the Equivalence Between the Maximin and Minimax SINR Problems}
We have thus far shown that the maximin problem \eqref{max-min-2018-general-rank-1} and the minimax problem \eqref{min-max-2018-general-rank-1} are equivalent. This main result can now be restated in the form of a saddle-point theorem as follows.

\begin{theorem}\label{maximin-minimax-equiv-thm}
Suppose that $\bar{\cal B}_0\subset\mathbb{C}^{N\times M}\setminus\{\bzero\}$ and $\bar{\cal B}_1\subset{\cal H}_{++}^N$ are convex and closed. Let $(\bQ^\star,\bR_1^\star)$ ibe an optimal solution of \eqref{min-max-2018-general-rank-1-inner-max-dual-equv-1}, and denote $\bw^\star=\bR_1^{\star -\frac{1}{2}}\bu_1^\star$, where $\bu_1^\star$ is an eigenvector associated with the largest eigenvalue $\lambda^\star=\lambda_1(\bR_1^{\star -\frac{1}{2}}\bQ^\star\bQ^{\star H}\bR_1^{\star -\frac{1}{2}})$ (namely, $(\lambda^\star,\bQ^\star,\bR_1^\star)$ is optimal for \eqref{min-max-2018-general-rank-1-inner-max-dual}). Then it holds that
\begin{equation}\label{maximin-minimax-equiv}
\begin{array}{l}
\bw^{\star H}\bQ^\star\bQ^{\star H}\bw^\star=\underset{\bw\ne\bzero}{\sf{maximize}}\underset{\bQ\in\bar{\cal
B}_0,\bR_1\in\bar{\cal B}_1}{\sf{minimize}}~\frac{\bw^H\bQ\bQ^H\bw}{\bw^H\bR_{1}\bw} \\
~~~~~~~~~~~~~~~~~~~~~~~~~~~~~= \underset{\bQ\in\bar{\cal
B}_0,\bR_1\in\bar{\cal B}_1}{\sf{minimize}}\underset{\bw\ne\bzero}{\sf{maximize}}~\frac{\bw^H\bQ\bQ^H\bw}{\bw^H\bR_{1}\bw} .
\end{array}
\end{equation}
Moreover, $(\bw^\star,\bQ^\star,\bR_1^\star)$ is an optimal solution for both the maximin and minimax SINR problems.
\end{theorem}

Remark that, by utilizing the main theorem, the worst-case SINR maximization problem \eqref{max-min-2018-general-rank-1} (i.e., the maximin SINR problem) can be solved by solving the SDP problem \eqref{min-max-2018-general-rank-1-inner-max-dual} in a single shot, provided that $\bar{\cal B}_0$ and $\bar{\cal B}_1$ are convex and closed. Indeed, suppose that $(\lambda^\star, \bQ^\star,\bR_1^\star)$ is an optimal solution of \eqref{min-max-2018-general-rank-1-inner-max-dual}, and thus $(\bQ^\star, \bR_1^\star)$ is optimal for \eqref{min-max-2018-general-rank-1-inner-max-dual-equv-1} with the optimal value $\lambda^\star=\lambda_1(\bR_1^{\star -\frac{1}{2}}\bQ^\star\bQ^{\star H}\bR_1^{\star -\frac{1}{2}})$. It follows from Theorem~\ref{maximin-minimax-equiv-thm} that $\lambda^\star$ is equal to the optimal value of maximin problem~\eqref{max-min-2018-general-rank-1} and that $(\bR_1^{\star -\frac{1}{2}}\bu_1^\star, \bQ^\star, \bR_1^\star)$, where $\bu_1^\star$ is an eigenvector associated with $\lambda_1(\bR_1^{\star -\frac{1}{2}} \bQ^\star \bQ^{\star H} \bR_1^{\star -\frac{1}{2}})$, constitutes an optimal solution of problem \eqref{max-min-2018-general-rank-1}.

%Since problem \eqref{min-max-2018-general-rank-1-inner-max-dual-equv-1} is equivalent to \eqref{min-max-2018-general-rank-1-inner-max-dual}, therefore $(\bQ^\star,\bR_1^\star)$ is optimal for \eqref{min-max-2018-general-rank-1-inner-max-dual-equv-1} with the optimal value $\lambda^\star$, if $(\lambda^\star,\bQ^\star,\bR_1^\star)$ is an optimal solution for \eqref{min-max-2018-general-rank-1-inner-max-dual}, and vice versa.
%Thereby, in order to get an optimal solution $(\bw^\star,\bQ^\star,\bR_1^\star)$ for the worst-case SINR maximization problem \eqref{max-min-2018-general-rank-1} (the maximin problem), we only have to solve \eqref{min-max-2018-general-rank-1-inner-max-dual}. Specifically, once problem \eqref{min-max-2018-general-rank-1-inner-max-dual} has been solved, getting an optimal solution $(\lambda^\star,\bQ^\star,\bR_1^\star)$, we claim that

Interestingly, when $\bQ\in\mathbb{C}^{N\times M}$ reduces to a vector $\ba\in\mathbb{C}^N$ (corresponding to the rank-one signal model), Theorem~\ref{maximin-minimax-equiv-thm} specializes to the following corollary.

\begin{corollary}\label{maximin-minimax-equiv-rank-1-coro}
Suppose that $\bar{\cal A}\subset\mathbb{C}^{N}\setminus\{\bzero\}$ and $\bar{\cal B}_1\subset{\cal H}_{++}^N$ are convex and closed. Let $(\ba^\star, \bR_1^\star)$ be an optimal solution of the following convex problem
\begin{equation}\label{min-max-2018-general-rank-1-inner-max-dual-equv-1-a}
\begin{array}[c]{cl}
\underset{\ba,\bR_1}{\sf{minimize}} & \ba^{H}\bR_1^{-1}\ba\\
\sf{subject\;to}           & \ba\in\bar{\cal A},\,\bR_1\in\bar{\cal B}_1,
\end{array}%\right.
\end{equation}
and let $\bw^\star=\bR_1^{\star -1}\ba^\star/\|\bR_1^{\star -\frac{1}{2}}\ba^\star\|$. Then it holds that
\begin{equation}\label{maximin-minimax-equiv-rank-1}
\begin{array}{l}
\bw^{\star H}\ba^\star\ba^{\star H}\bw^\star=\underset{\bw\ne\bzero}{\sf{maximize}}\underset{\ba\in\bar{\cal
A},\bR_1\in\bar{\cal B}_1}{\sf{minimize}}~\frac{\bw^H\ba\ba^H\bw}{\bw^H\bR_{1}\bw} \\
~~~~~~~~~~~~~~~~~~~~~~~~~~~~~~~~~~= \underset{\ba\in\bar{\cal
A},\bR_1\in\bar{\cal B}_1}{\sf{minimize}}\underset{\bw\ne\bzero}{\sf{maximize}}~\frac{\bw^H\ba\ba^H\bw}{\bw^H\bR_{1}\bw}.
\end{array}
\end{equation}
Moreover, $(\bw^\star, \ba^\star,\bR_1^\star)$ is an optimal solution for the maximin SINR and minimax SINR problems.
\end{corollary}

We remark that this corollary is slightly different from Theorem~1 in \cite{Kim06-bc}, and the difference lies in the assumptions on the underlying sets $\bar{\cal A}$ and $\bar{\cal B}_1$: here they are required to be convex and closed, whereas they are assumed to be convex and compact in \cite{Kim06-bc}. Furthermore, Theorem~\ref{maximin-minimax-equiv-thm} can be equivalently reformulated as an explicit saddle-point property of the SINR for the general-rank signal model, as stated below.

\begin{theorem}\label{maximin-minimax-equiv-saddle-point-prop}
Suppose that $\bar{\cal B}_0\subset\mathbb{C}^{N\times M}\setminus\{\bzero\}$ and $\bar{\cal B}_1\subset{\cal H}_{++}^N$ are convex and closed. Let $(\bQ^\star, \bR_1^\star)$ be an optimal solution for \eqref{min-max-2018-general-rank-1-inner-max-dual-equv-1} and $\bw^\star=\bR_1^{\star -\frac{1}{2}}\bu_1^\star$, where $\bu_1^\star$ is an eigenvector associated with the largest eigenvalue $\lambda_1(\bR_1^{\star -\frac{1}{2}}\bQ^\star\bQ^{\star H}\bR_1^{\star -\frac{1}{2}})$. Define the SINR as
\begin{equation}\label{define-SINR-function}
S(\bw,\bQ,\bR_1)=\frac{\bw^H\bQ\bQ^H\bw}{\bw^H\bR_{1}\bw}.
\end{equation}
Then $(\bw^\star, \bQ^\star, \bR_1^\star)$ is a saddle point of the SINR function, satisfying
\begin{equation}\label{maximin-minimax-equiv-saddle-point}
\begin{array}{l}
S(\bw,\bQ^\star,\bR_1^\star)\le S(\bw^\star,\bQ^\star,\bR_1^\star)\le S(\bw^\star,\bQ,\bR_1)
%\frac{\bw^{\star H}\bQ^\star\bQ^{\star H}\bw^\star}{\bw^{\star H}\bR_{1}^\star\bw^\star}\le\frac{\bw^{\star H}\bQ\bQ^H\bw^\star}{\bw^{\star H}\bR_{1}\bw^\star}
\end{array}
\end{equation}
for $\bw\ne\bzero$, $\bQ\in\bar{\cal B}_0$ and $\bR_1\in\bar{\cal B}_1$.
\end{theorem}

\section{Solving SDP Problem \eqref{min-max-2018-general-rank-1-inner-max-dual} with Some Specific Convex and Closed Sets $\bar{\cal B}_0$ and $\bar{\cal B}_1$}

In this section, we examine how to solve the SDP problem \eqref{min-max-2018-general-rank-1-inner-max-dual} for various choices of convex and closed uncertainty sets $\bar{\cal B}_0$ and $\bar{\cal B}_1$. As noted earlier, by solving the SDP problem \eqref{min-max-2018-general-rank-1-inner-max-dual}, we obtain a globally optimal solution for both the maximin SINR problem \eqref{max-min-2018-general-rank-1} and the minimax SINR problem \eqref{min-max-2018-general-rank-1} simultaneously, provided that the uncertainty sets $\bar{\cal B}_0 \subset \mathbb{C}^{N \times M} \setminus \{\mathbf{0}\}$ and $\bar{\cal B}_1 \subset \mathcal{H}_{++}^N$ are convex and closed.

\subsection{Frobenius Norm Ball and the Spectral Norm Ball of the Matrix Errors}

\subsubsection{Frobenius Norm Ball of the Matrix Errors}
Let $\bar{\cal B}_0 = \tilde{\cal B}_0$ and $\bar{\cal B}_1 = \tilde{\cal B}_1$ be defined as in \eqref{example-B0-bar-1} and \eqref{example-B1-bar}, respectively. Clearly, they are convex, closed and LMI representable. Then, problem \eqref{min-max-2018-general-rank-1-inner-max-dual} is specified to
\begin{equation}\label{min-max-2018-general-rank-1-inner-max-dual-B0-bar-B1-bar-1}
\begin{array}[c]{cl}
\underset{\lambda,\bQ,\bR_1}{\sf{minimize}} & \lambda\\
\sf{subject\;to}           & \left[\begin{array}{cc}\bR_{1}&\bQ\\ \bQ^H&\lambda\bI\end{array}\right]\succeq\bzero\\
                           & \|\bQ-\hat\bQ\|_F^2\le\eta\\
                           & \|\bR_1-\hat\bR\|_F^2\le\gamma_1.
\end{array}%\right.
\end{equation}
It can be verified that problems \eqref{max-min-2018-general-rank-3} and \eqref{min-max-2018-general-rank-1-inner-max-dual-B0-bar-B1-bar-1} are equivalent. Indeed, problem \eqref{max-min-2018-general-rank-3} is identical to the maximin problem \eqref{max-min-2018-general-rank-1} (see, e.g., \cite{KV13-tsp,HV18-spl}), and problem \eqref{min-max-2018-general-rank-1-inner-max-dual-B0-bar-B1-bar-1} is also equivalent to \eqref{max-min-2018-general-rank-1} by Theorem~\ref{maximin-minimax-equiv-thm} and Proposition~\ref{min-max-convex-opt-problem}. Therefore, solving problem \eqref{min-max-2018-general-rank-1-inner-max-dual-B0-bar-B1-bar-1} yields an optimal robust adaptive beamformer $\bw^\star$ for problem \eqref{max-min-2018-general-rank-3}.

Specifically, suppose that $(\lambda^\star,\bQ^\star,\bR_1^\star)$ is an optimal solution of SDP problem \eqref{min-max-2018-general-rank-1-inner-max-dual-B0-bar-B1-bar-1}. Then, $(\bw^\star,\bQ^\star,\bR_1^\star)$ s an optimal solution of the maximin problem \eqref{max-min-2018-general-rank-1}, where $\bw^\star=\bR_1^{\star -\frac{1}{2}}\bu_1^\star$ and $\bu_1^\star$ is an eigenvector associated with $\lambda^\star=\lambda_1(\bR_1^{\star -\frac{1}{2}}\bQ^\star\bQ^{\star H}\bR_1^{\star -\frac{1}{2}})$.

We recall that problem \eqref{max-min-2018-general-rank-3} can be solved using an iterative SOCP-based approximation method (see, e.g., \cite{HV18-spl}). Specifically, since the objective function in \eqref{max-min-2018-general-rank-3} is a difference-of-convex (DC) function, the method proceeds by linearizing the first convex term in the objective at each iteration. In practice, the total computational time of this approximate algorithm is often lower than that required to solve problem \eqref{min-max-2018-general-rank-1-inner-max-dual-B0-bar-B1-bar-1}, particularly when the number of array antennas exceeds approximately ten. This is because each iteration only requires solving a relatively lightweight SOCP, and the total number of iterations is typically small (around four) and largely insensitive to the choice of the initial point $\bw_0$.

In contrast, problem \eqref{min-max-2018-general-rank-1-inner-max-dual-B0-bar-B1-bar-1} involves an LMI constraint of size $(M+N)\times(M+N)$, along with two high-dimensional second-order cone (SOC) constraints, namely, $\|\vect(\bQ)-\vect(\hat\bQ)\|\le\sqrt{\eta}$ and $\|\vect(\bR_1)-\vect(\hat\bR)\|\le\sqrt{\gamma_1}$). As a result, solving this LMI-based SDP using an interior-point method entails a significantly higher computational cost, despite the fact that the problem needs to be solved only once, as will be demonstrated in the simulation results.

\subsubsection{Spectral Norm Ball of the Matrix Errors}

%\begin{equation}\label{example-B0-bar-two-norm}
%\hat{\cal B}_0=\{\bQ\in\mathbb{C}^{N\times M}~|~\|\bQ-\hat\bQ\|_2^2\le\eta\},
%\end{equation}
%and
%\begin{equation}\label{example-B1-bar-two-norm}
%\hat{\cal B}_1=\{\bR_1\in\mathbb{C}^{N\times N}~|~\|\bR_1-\hat\bR\|_2^2\le\gamma_1,\,\bR_{1}\succeq\bzero \}.
%\end{equation}
Note that the worst-case SINR maximization problem \eqref{max-min-2018-general-rank-1} with the uncertainty sets $\hat{\cal B}_0$ and $\hat{\cal B}_1$ (as in case of $\tilde{\cal B}_0$ and $\tilde{\cal B}_1$) is always equivalent to problem \eqref{max-min-2018-general-rank-3}. On the other hand, it follows from Proposition~\ref{min-max-convex-opt-problem} and Theorem~\ref{maximin-minimax-equiv-thm} that the maximin problem \eqref{max-min-2018-general-rank-1} with uncertainty sets $\hat{\cal B}_0$ and $\hat{\cal B}_1$ is identical to the following LMI problem:
\begin{equation}\label{min-max-2018-general-rank-1-inner-max-dual-B0-bar-B1-bar-two-norm}
\begin{array}[c]{cl}
\underset{\lambda,\bQ,\bR_1}{\sf{minimize}} & \lambda\\
\sf{subject\;to}           & \left[\begin{array}{cc}\bR_{1}&\bQ\\ \bQ^H&\lambda\bI\end{array}\right]\succeq\bzero\\
                           & \|\bQ-\hat\bQ\|_2^2\le\eta\\
                           & \|\bR_1-\hat\bR\|_2^2\le\gamma_2.
\end{array}%\right.
\end{equation}

In other words, problem \eqref{max-min-2018-general-rank-3} is equivalent to problem \eqref{min-max-2018-general-rank-1-inner-max-dual-B0-bar-B1-bar-two-norm}. Consequently, the three problems \eqref{max-min-2018-general-rank-3}, \eqref{min-max-2018-general-rank-1-inner-max-dual-B0-bar-B1-bar-two-norm}, and \eqref{min-max-2018-general-rank-1-inner-max-dual-B0-bar-B1-bar-1} are equivalent to one another. Hence, problem \eqref{max-min-2018-general-rank-3} can be solved by solving \eqref{min-max-2018-general-rank-1-inner-max-dual-B0-bar-B1-bar-two-norm}. In particular, the solution $\bw^\star=\bR_1^{\star -\frac{1}{2}}\bu_1^\star$ is optimal for problem \eqref{max-min-2018-general-rank-3} too, where $\bu_1^\star$ is a principal eigenvector of $\bR_1^{\star -\frac{1}{2}}\bQ^\star\bQ^{\star H}\bR_1^{\star -\frac{1}{2}}$ and $\bQ^\star, \bR_1^\star$ are components for an optimal solution for SDP problem \eqref{min-max-2018-general-rank-1-inner-max-dual-B0-bar-B1-bar-two-norm}.

However, the computational cost of SDP problem \eqref{min-max-2018-general-rank-1-inner-max-dual-B0-bar-B1-bar-two-norm} is higher than that of \eqref{min-max-2018-general-rank-1-inner-max-dual-B0-bar-B1-bar-1}. This is because problem \eqref{min-max-2018-general-rank-1-inner-max-dual-B0-bar-B1-bar-two-norm} can be recast as (see, e.g., \cite[Example~19, p.~152]{Nemi-book-2001} and \cite[Example~18b, p.~147]{Nemi-book-2001})
\begin{equation}\label{min-max-2018-general-rank-1-inner-max-dual-B0-bar-B1-bar-two-norm-sdr}
\begin{array}[c]{cl}
\underset{\lambda,\bQ,\bR_1}{\sf{minimize}} & \lambda\\
\sf{subject\;to}           & \left[\begin{array}{cc}\bR_{1}&\bQ\\ \bQ^H&\lambda\bI\end{array}\right]\succeq\bzero\\
                           & \left[\begin{array}{cc}\eta\bI& \bQ-\hat\bQ\\ \bQ^H-\hat\bQ^H&\bI \end{array}\right]\succeq\bzero \\ %\|\bQ-\hat\bQ\|_2^2\le\eta\\
                           & \sqrt{\gamma_2}\bI\succeq\bR_1-\hat\bR\succeq-\sqrt{\gamma_2}\bI, %\|\bR_1-\hat\bR_1\|_2^2\le\gamma_1.
\end{array}%\right.
\end{equation}
which includes four LMI constraints.

\subsection{Additional Energy Constraint in the Uncertainty Set of the INC Matrix $\bR_1$}
In this subsection, we consider the following convex and closed uncertainty sets of the INC matrix $\bR_1$:
\begin{equation}\label{example-B0-bar-1-tr-F}
\tilde{\cal B}_1^\prime=\{\bR_1\in{\cal H}^{N}~|~\|\bR_1-\hat\bR\|_F^2\le\gamma_1,\,\rho_1\le\tr\bR_1\le\rho_2,\,\bR_1\succeq\bzero\},
\end{equation}
and
\begin{equation}\label{example-B0-bar-1-tr-2}
\hat{\cal B}_1^{\prime}=\{\bR_1\in{\cal H}^{N}~|~\|\bR_1-\hat\bR\|_2^2\le\gamma_2,\,\rho_1\le\tr\bR_1\le\rho_2,\,\bR_1\succeq\bzero\},
\end{equation}
where the double-sided constraint on the trace of $\mathbf{R}_1$ specifies a trust interval for the interference-plus-noise energy (cf. \cite{ADHP-16TSP}).

\subsubsection{Solving Maximin Problem \eqref{max-min-2018-general-rank-1} and Its Minimax Counterpart Problem~\eqref{min-max-2018-general-rank-1} with the INC Uncertainty Set $\bar{\cal B}_1^\prime$}

With the new INC uncertainty set \eqref{example-B0-bar-1-tr-F}, the corresponding worst-case SINR maximization problem \eqref{max-min-2018-general-rank-1} can reformulated as
\begin{equation}\label{max-min-2018-general-rank-1-Bprime}
\begin{array}[c]{cl}
\underset{\bw\ne\bzero}{\sf{maximize}} &\underset{\bQ\in\tilde{\cal
B}_0,\bR_1\in\tilde{\cal B}_1^\prime}{\sf{minimize}}
\begin{array}[c]{c}\displaystyle\frac{\bw^H\bQ\bQ^H\bw}{\bw^H\bR_{1}\bw},
\end{array}
\end{array}%\right.
\end{equation}
and its minimax counterpart problem \eqref{min-max-2018-general-rank-1} can be rewritten as
\begin{equation}\label{min-max-2018-general-rank-1-Bprime}
\begin{array}[c]{cl}
\underset{\bQ\in\tilde{\cal
B}_0,\bR_1\in\tilde{\cal B}_1^\prime}{\sf{minimize}} & \underset{\bw\ne\bzero}{\sf{maximize}}
\begin{array}[c]{c}\displaystyle\frac{\bw^H\bQ\bQ^H\bw}{\bw^H\bR_{1}\bw}.
\end{array}
\end{array}%\right.
\end{equation}
It follows from Proposition~\ref{min-max-convex-opt-problem} and Theorem~\ref{maximin-minimax-equiv-thm} that the maximin problem \eqref{max-min-2018-general-rank-1-Bprime} and the minimax problem \eqref{min-max-2018-general-rank-1-Bprime} are equivalent to each other, and they can be solved via solving the following SDP problem
\begin{equation}\label{min-max-2018-general-rank-1-inner-max-dual-Bprime}
\begin{array}[c]{cl}
\underset{\lambda,\bQ,\bR_1}{\sf{minimize}} & \lambda\\
\sf{subject\;to}           & \left[\begin{array}{cc}\bR_{1}&\bQ\\ \bQ^H&\lambda\bI\end{array}\right]\succeq\bzero\\
                           & \|\bQ-\hat\bQ\|_F^2\le\eta\\
                           & \|\bR_1-\hat\bR\|_F^2\le\gamma_1\\
                           & \rho_1\le\tr\bR_1\le\rho_2.
\end{array}%\right.
\end{equation}
Suppose that $(\lambda^\star, \bQ^\star, \bR_1^\star)$ is an optimal solution for \eqref{min-max-2018-general-rank-1-inner-max-dual-Bprime}. Then, $(\bw^\star, \bQ^\star, bR_1^\star)$ is optimal for the maximin problem \eqref{max-min-2018-general-rank-1-Bprime} and the minimax problem \eqref{min-max-2018-general-rank-1}, where $\bw^\star=\bR_1^{\star -\frac{1}{2}}\bu_1^\star$ with $\bu_1^\star$ being an eigenvector associated with $\lambda^\star=\lambda_1(\bR_1^{\star -\frac{1}{2}}\bQ^\star\bQ^{\star H}\bR_1^{\star -\frac{1}{2}})$.

On the other hand, the maximin problem \eqref{max-min-2018-general-rank-1-Bprime} can be equivalently transformed into:
\begin{equation}\label{max-min-2018-general-rank-3-Bprime}
\begin{array}[c]{cl}\underset
{\bw\ne\bzero}{\sf{maximize}} & \|\hat\bQ^H\bw\|-\sqrt{\eta}\|\bw\|\\
\sf{subject\;to}               &
\underset{\bR_1\in\tilde{\cal B}_1^\prime}{\sf{maximize}} \quad\bw^H\bR_1\bw \le1.
\end{array}%\right.
\end{equation}
Note that the maximization problem in the constraint of problem \eqref{max-min-2018-general-rank-3-Bprime} is a convex optimization problem, and its dual can be formulated as follows
\begin{equation}\label{max-min-2018-general-rank-3-Bprime-dual}
\begin{array}[c]{cl}
\underset{x,y,\bX}{\sf{minimize}} & \sqrt{\gamma_1}\|\bX\|_F -\tr(\hat\bR\bX)-\rho_1 x +\rho_2 y\\
\sf{subject\;to}           & \left[\begin{array}{cc}(y-x)\bI-\bX&\bw\\ \bw^H&1\end{array}\right]\succeq\bzero\\
                           & y\ge0, x\ge0.
\end{array}%\right.
\end{equation}
Accordingly, problem \eqref{max-min-2018-general-rank-3-Bprime} can be recast into
\begin{equation}\label{max-min-2018-general-rank-3-Bprime-equiv}
\begin{array}[c]{cl}\underset
{x,y,\bw,\bX}{\sf{maximize}} & \|\hat\bQ^H\bw\|-\sqrt{\eta}\|\bw\|\\
\sf{subject\;to}               & \sqrt{\gamma_1}\|\bX\|_F -\tr(\hat\bR\bX)-\rho_1 x +\rho_2 y \le1\\
                                &\left[\begin{array}{cc}(y-x)\bI-\bX&\bw\\ \bw^H&1\end{array}\right]\succeq\bzero\\
                                &y\ge0, x\ge0.
\end{array}%\right.
\end{equation}
Since the objective function of problem \eqref{max-min-2018-general-rank-3-Bprime-equiv} is a difference of convex functions, an iterative approximation algorithm can be applied by linearizing the first convex term (see, e.g., \cite{HV18-spl,SBP-2017tsp}). For example, the objective function can be linearized as
\begin{equation}\label{obj-DC-linearization-1}
\frac{\Re(\bw_k^H\hat\bQ\hat\bQ^H\bw)}{\|\hat\bQ^H\bw_k\|}-\sqrt{\eta}\|\bw\|,
\end{equation}
where $\bw_k$ is obtained from the $(k-1)$th iteration.
In practice, however, the computational cost of solving problem \eqref{max-min-2018-general-rank-3-Bprime-equiv} is higher than that of problem \eqref{min-max-2018-general-rank-1-inner-max-dual-Bprime}, since multiple iterative steps are required to solve the associated SDP problem, namely,
\begin{equation}\label{max-min-2018-general-rank-3-Bprime-equiv-lin}
\begin{array}[c]{cl}\underset
{x,y,\bw,\bX}{\sf{maximize}} & \frac{\Re(\bw_k^H\hat\bQ\hat\bQ^H\bw)}{\|\hat\bQ^H\bw_k\|}-\sqrt{\eta}\|\bw\|\\
\sf{subject\;to}               & \sqrt{\gamma_1}\|\bX\|_F -\tr(\hat\bR\bX)-\rho_1 x +\rho_2 y \le1\\
                                &\left[\begin{array}{cc}(y-x)\bI-\bX&\bw\\ \bw^H&1\end{array}\right]\succeq\bzero\\
                                &y\ge0, x\ge0,
\end{array}%\right.
\end{equation}
to converge to an optimal solution for \eqref{max-min-2018-general-rank-3-Bprime-equiv}, whereas the SDP problem \eqref{min-max-2018-general-rank-1-inner-max-dual-Bprime} is solved in a single shot.
%In other word, suppose that $\bw_k$ is the output of the $(k-1)$th step for $k\ge1$. Observe that
%\begin{equation}\label{linearization-of-norm}
%\|\hat\bQ^H\bw\|\ge\frac{\bw_k^H\hat\bQ\hat\bQ^H\bw}{\|\hat\bQ^H\bw_k\|}\ge\frac{\Re(\bw_k^H\hat\bQ\hat\bQ^H\bw)}{\|\hat\bQ^H\bw_k\|};
%\end{equation}
%then we have to solve the following convex approximation problem in the $k$th step:
%\begin{equation}\label{max-min-2018-general-rank-3-Bprime-equiv-sca}
%\begin{array}[c]{cl}\underset
%{x,y,\bw,\bX}{\sf{maximize}} & \frac{\Re(\bw_k^H\hat\bQ\hat\bQ^H\bw)}{\|\hat\bQ^H\bw_k\|}-\sqrt{\eta}\|\bw\|\\
%\sf{subject\;to}               & \sqrt{\gamma_1}\|\bX\|_F-\rho_1 x +\rho_2 y -\tr(\hat\bR\bX) \le1\\
%                                &\left[\begin{array}{cc}(y-x)\bI-\bX&\bw\\ \bw^H&1\end{array}\right]\succeq\bzero\\
%                                &x\ge0,\,y\ge0.
%\end{array}%\right.
%\end{equation}
%Clearly, the sequence of the optimal values for \eqref{max-min-2018-general-rank-3-Bprime-equiv-sca} is ascending,

\subsubsection{Solving the Maximin Problem~\eqref{max-min-2018-general-rank-1} and Its Minimax Counterpart Problem~\eqref{min-max-2018-general-rank-1} with the INC Uncertainty Set $\hat{\cal B}_1^{\prime}$}

In this case, the maximin problem \eqref{max-min-2018-general-rank-1} is written as
\begin{equation}\label{max-min-2018-general-rank-1-Bhat-prime}
\begin{array}[c]{cl}
\underset{\bw\ne\bzero}{\sf{maximize}} &\underset{\bQ\in\hat{\cal
B}_0,\bR_1\in\hat{\cal B}_1^{\prime}}{\sf{minimize}}
\begin{array}[c]{c}\displaystyle\frac{\bw^H\bQ\bQ^H\bw}{\bw^H\bR_{1}\bw},
\end{array}
\end{array}%\right.
\end{equation}
and its equivalent minimax problem \eqref{min-max-2018-general-rank-1} is expressed as
\begin{equation}\label{min-max-2018-general-rank-1-Bhat-prime}
\begin{array}[c]{cl}
\underset{\bQ\in\hat{\cal
B}_0,\bR_1\in\hat{\cal B}_1^{\prime}}{\sf{minimize}} & \underset{\bw\ne\bzero}{\sf{maximize}}
\begin{array}[c]{c}\displaystyle\frac{\bw^H\bQ\bQ^H\bw}{\bw^H\bR_{1}\bw}.
\end{array}
\end{array}%\right.
\end{equation}
Clearly, the minimax problem \eqref{min-max-2018-general-rank-1-Bhat-prime} can be equivalently recast as
\begin{equation}\label{min-max-2018-general-rank-1-inner-max-dual-B0-hat-B1-hat-two-norm-prime}
\begin{array}[c]{cl}
\underset{\lambda,\bQ,\bR_1}{\sf{minimize}} & \lambda\\
\sf{subject\;to}           & \left[\begin{array}{cc}\bR_{1}&\bQ\\ \bQ^H&\lambda\bI\end{array}\right]\succeq\bzero\\
                           & \|\bQ-\hat\bQ\|_2^2\le\eta\\
                           & \|\bR_1-\hat\bR\|_2^2\le\gamma_2\\
                           & \rho_1\le\tr\bR_1\le\rho_2,
\end{array}%\right.
\end{equation}
which is an SDP problem. Then, an optimal $\bw^\star$ is obtained in a way similar to \eqref{min-max-2018-general-rank-1-inner-max-dual-Bprime}.

The maximin problem \eqref{max-min-2018-general-rank-1-Bhat-prime} can be equivalently reformulated as the following nonconvex problem
\begin{equation}\label{max-min-2018-general-rank-3-B-hat-prime}
\begin{array}[c]{cl}\underset
{\bw\ne\bzero}{\sf{maximize}} & \|\hat\bQ^H\bw\|-\sqrt{\eta}\|\bw\|\\
\sf{subject\;to}               &
\underset{\bR_1\in\hat{\cal B}_1^\prime}{\sf{maximize}} \quad\bw^H\bR_1\bw \le1.
\end{array}%\right.
\end{equation}
To tackle this problem, we first examine the dual of the convex maximization problem appearing in the constraint of \eqref{max-min-2018-general-rank-3-B-hat-prime}. Specifically, the dual problem can be derived as follows
\begin{equation}\label{max-min-2018-general-rank-3-Bprime-dual-2norm}
\begin{array}[c]{cl}
\underset{x,y,\bX,\bY}{\sf{minimize}} & \sqrt{\gamma_2}\tr(\bX+\bY) + \tr(\hat\bR(\bY-\bX))-\rho_1 x +\rho_2 y\\
\sf{subject\;to}           & \left[\begin{array}{cc}(y-x)\bI+\bY-\bX&\bw\\ \bw^H&1\end{array}\right]\succeq\bzero\\
                           & x\ge0,\,y\ge0,\,\bX\succeq\bzero,\,\bY\succeq\bzero.
\end{array}%\right.
\end{equation}
where the terms involving $\rho_2$ and $\rho_1$ are canceled.
Consequently, problem \eqref{max-min-2018-general-rank-3-B-hat-prime} can be equivalently transformed into
\begin{equation}\label{max-min-2018-general-rank-3-Bprime-equiv-2norm}
\begin{array}[c]{cl}\underset
{\bw,x,y,\bX,\bY}{\sf{maximize}} & \|\hat\bQ^H\bw\|-\sqrt{\eta}\|\bw\|\\
\sf{subject\;to}               & \sqrt{\gamma_2}\tr(\bX+\bY) + \tr(\hat\bR(\bY-\bX))-\rho_1 x +\rho_2 y \le1\\
                                & \left[\begin{array}{cc}(y-x)\bI+\bY-\bX&\bw\\ \bw^H&1\end{array}\right]\succeq\bzero \\
                                & x\ge0,\,y\ge0,\,\bX\succeq\bzero,\,\bY\succeq\bzero.
\end{array}%\right.
\end{equation}
Thus, an iterative algorithm -- similar to that used for solving \eqref{max-min-2018-general-rank-3-Bprime-equiv} -- can be applied to tackle the DC problem \eqref{max-min-2018-general-rank-3-Bprime-equiv-2norm}, where in the $k$th iteration the concave part of the objective function \eqref{obj-DC-linearization-1} is maximized subject to the constraints of \eqref{max-min-2018-general-rank-3-Bprime-equiv-2norm}. Naturally, the computational cost of solving \eqref{max-min-2018-general-rank-3-Bprime-equiv-2norm} is higher than that of problem \eqref{min-max-2018-general-rank-1-inner-max-dual-B0-hat-B1-hat-two-norm-prime}, even though the performance of the resulting beamformers from the two problems depends on the specific choice of the uncertainty sets for $\bQ$ and $\bR_1$.

\section{Numerical Experiments}

Consider a uniform linear array (ULA) with $N=10$ omni-directional sensors spaced half a wavelength apart. The array noise is spatially and temporally white Gaussian with zero mean and covariance matrix $\bI$. The actual desired signal is assumed to be locally incoherently scattered with a Gaussian angular power density, centered at $30^\circ$ with an angular spread of $4^\circ$, whereas the presumed desired signal has the same type of distribution but is centered at $34^\circ$ with an angular spread of $6^\circ$. An interferer impinges on the array with an interference-to-noise ratio (INR) of 30~dB, following a uniform angular power density with central angle $10^\circ$ and angular spread $10^\circ$ \cite{KV13-tsp}. The training sample size is set to $T=50$, and each data point in the simulation results is obtained by averaging over 200 independent runs.

We evaluate the performance of the beamformer obtained from the minimax SINR equivalent SDP problem \eqref{min-max-2018-general-rank-1-inner-max-dual} and compare it with beamformers computed using the iterative approaches in \cite{KV13-tsp, HV18-spl}, where the worst-case SINR maximization problem (i.e., the maximin SINR problem \eqref{max-min-2018-general-rank-3}) is solved iteratively. In the figures, the three beamformers are labeled as ``Minimax SINR," ``Maximin SINR KV,"  and ``Maximin SINR HV," respectively.

For these simulations, the uncertainty sets of $\bR_1$ and $\bQ$ are defined in \eqref{example-B1-bar} and \eqref{example-B0-bar-1}, with parameters set as $\sqrt{\gamma_1} = 0.1|\hat\bR|_F$ and $\sqrt{\eta} = 0.5|\hat\bQ|_F$, where $\hat\bR_s = \hat\bQ\hat\bQ^H$. In addition, we investigate the beamformer performance under a modified minimax SINR problem \eqref{min-max-2018-general-rank-1-inner-max-dual-Bprime}, which includes a double-sided energy constraint $\rho_1 \le \tr\bR_1 \le \rho_2$ in the uncertainty set (resulting in the new convex and closed set \eqref{example-B0-bar-1-tr-F}). For comparison, the corresponding maximin SINR problem \eqref{max-min-2018-general-rank-1-Bprime} is also solved using the iterative methods in \cite{KV13-tsp, HV18-spl}. Here, $\rho_1 = 0.5 \tr \hat\bR$ and $\rho_2 = 0.9 \tr \hat\bR$, where the coefficients are chosen below one because $\hat\bR$ represents the sample covariance matrix that includes the desired signal.

\begin{figure}[!h]
\centerline{\resizebox{.5\textwidth}{!}{\includegraphics{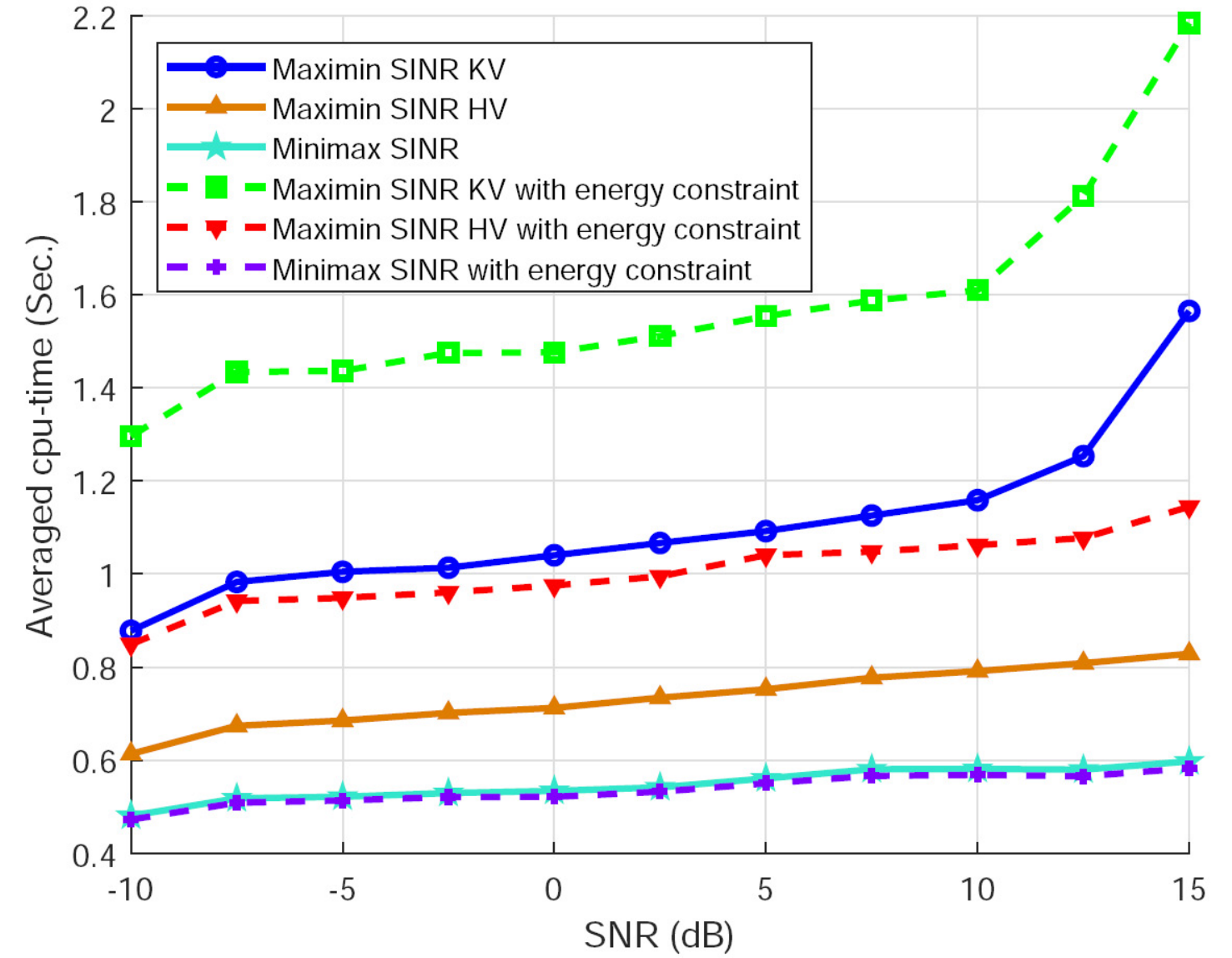}}
}
%\vspace*{-.5\baselineskip}
\caption{Averaged cpu-running time versus SNR, with INR =30~dB, $\rho_1=0.5\tr\hat\bR$, $\rho_2=0.9\tr\hat\bR$ and $T=50$, over 200 simulation runs.}
\label{fig-cpu-time}
%\vspace*{0\baselineskip}
\end{figure}

\begin{figure}[!h]
\centerline{\resizebox{.5\textwidth}{!}{\includegraphics{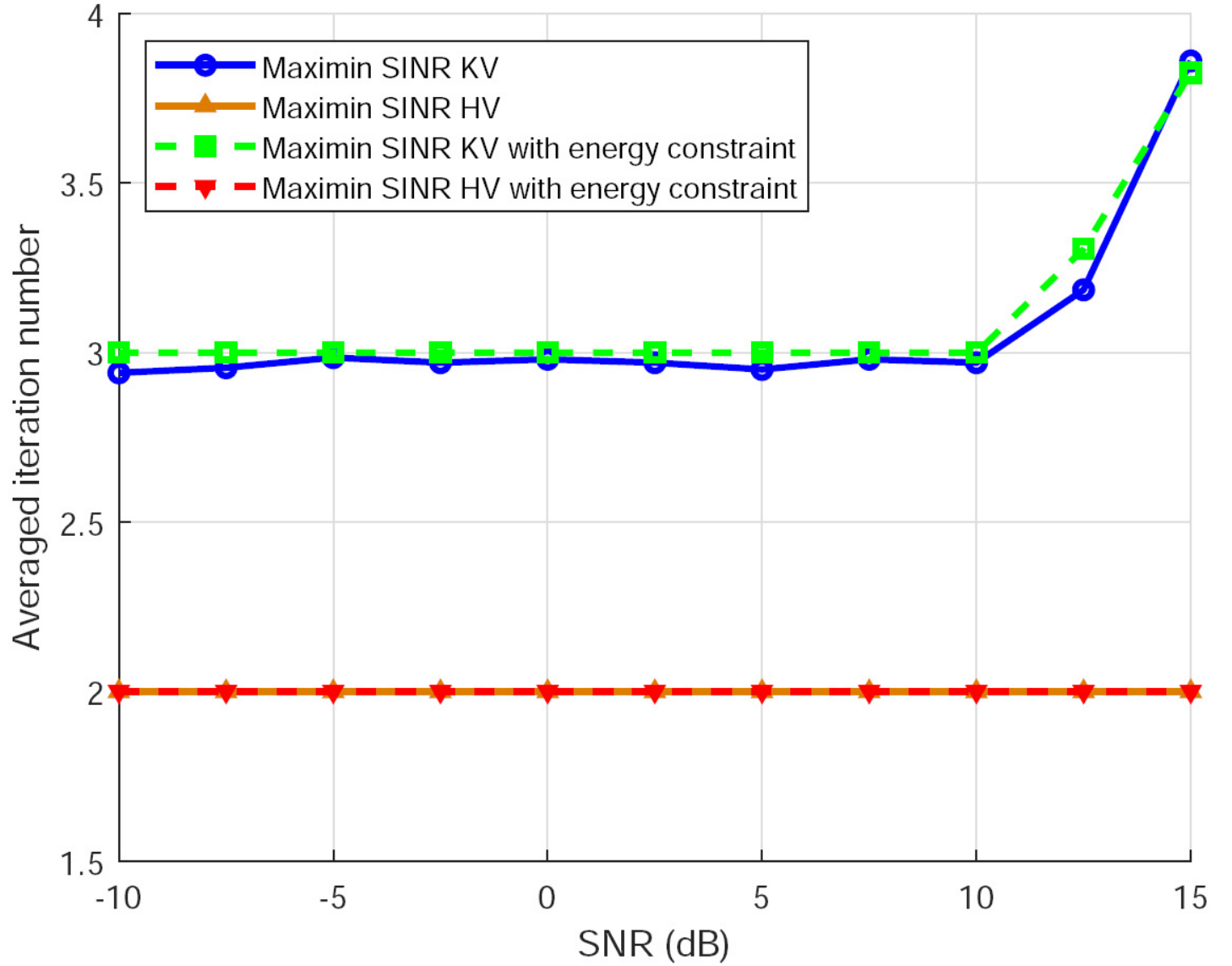}}
}
%\vspace*{-.5\baselineskip}
\caption{Averaged iteration number versus SNR, with INR =30~dB, $\rho_1=0.5\tr\hat\bR$, $\rho_2=0.9\tr\hat\bR$ and $T=50$, over 200 simulation runs.}
\label{fig-iter-number}
%\vspace*{0\baselineskip}
\end{figure}

Fig.~\ref{fig-cpu-time} shows the computational times required to compute the beamformers. Since the minimax SINR equivalent SDP problems \eqref{min-max-2018-general-rank-1-inner-max-dual} (without the energy constraint) and \eqref{min-max-2018-general-rank-1-inner-max-dual-Bprime} (with the energy constraint) are solved in a single step, their computational times are observed to be roughly the same.\footnote{The CPU is Intel Core i9-12900KF with 64~GB RAM on a desktop running the Matlab codes.} This is reasonable because the additional linear energy constraint in \eqref{min-max-2018-general-rank-1-inner-max-dual-Bprime} contributes minimally to the overall computation, which is dominated by handling the PSD and high-dimensional SOC constraints.
Moreover, solving both minimax SINR SDP problems is faster than solving the maximin SINR problems using the iterative algorithms in \cite{KV13-tsp, HV18-spl}, whether or not the energy constraint is included. This is expected, as each minimax SDP is solved in a single short, whereas the approximate algorithms for the maximin problem are iterative, requiring either an SDP or an SOCP to be solved at each iteration.
In addition, the computational cost of the SDP-based iterative algorithm in \cite{KV13-tsp} is higher than that of the SOCP-based iterative procedure in \cite{HV18-spl}, consistent with the observations in \cite{HV18-spl}. Finally, including the energy constraint increases the computation time for both the SDP- and SOCP-based iterative methods, since the size of the SDP or SOCP at each iteration becomes larger. Similar behavior is observed regardless of whether the energy constraint is present or not.

Fig.~\ref{fig-iter-number} illustrates that the average number of iterations still converges under the same settings. Specifically, the maximin SINR problem without (or with) the linear energy constraint, solved using the (similar) SOCP-based approximation algorithm in \cite{HV18-spl}, requires only about two iterations. In contrast, the SDP-based approximation approach in \cite{KV13-tsp} takes around three to four iterations for the maximin SINR problems without (or with) the energy constraint. This observation is consistent with the computational times reported in Fig.~\ref{fig-cpu-time}.

% In words, the objective function values are sufficiently close by using the same threshold and the same initial point is provided in \cite{KV13-tsp}.

\begin{figure}[!h]
\centerline{\resizebox{.51\textwidth}{!}{\includegraphics{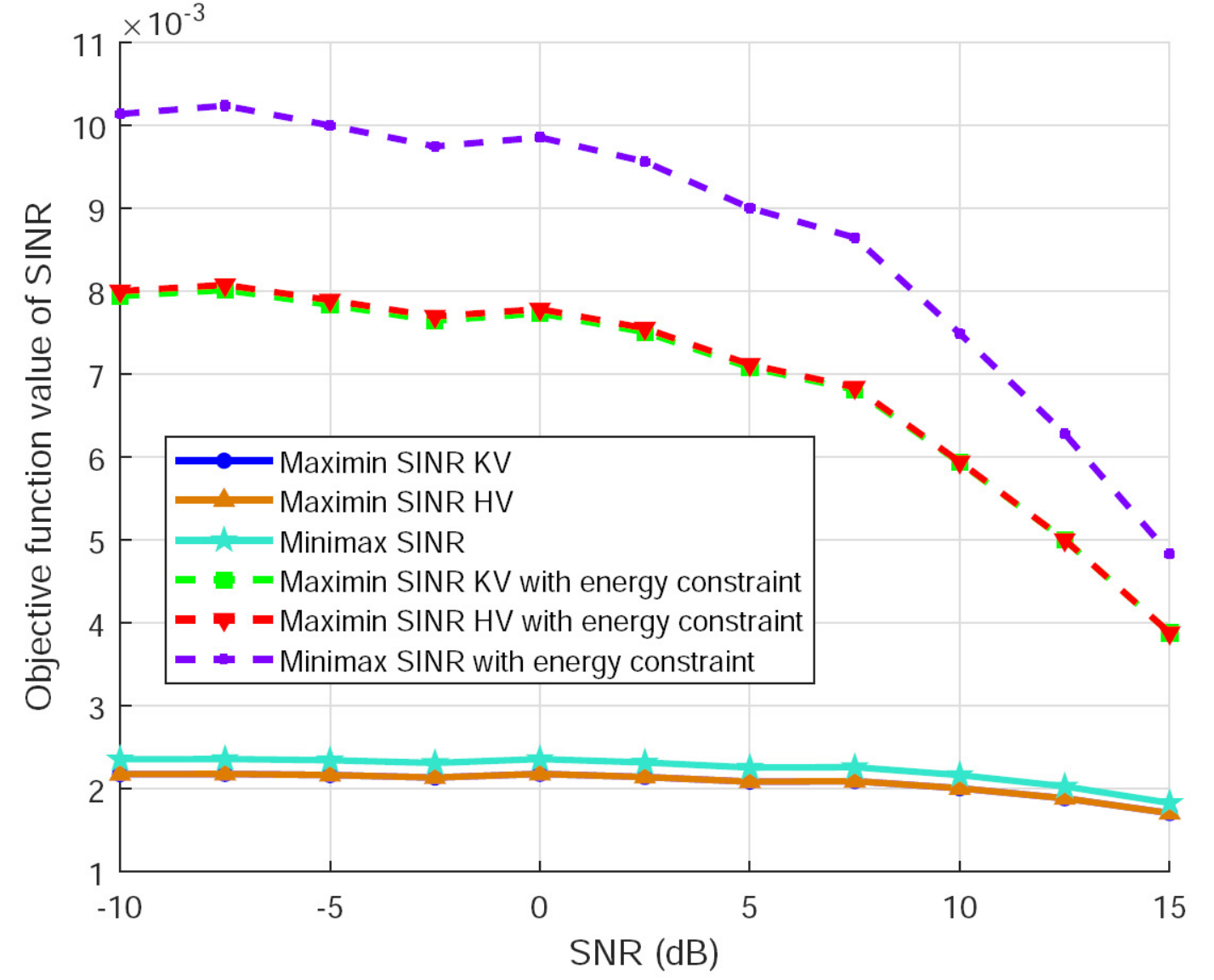}}
}
%\vspace*{-.5\baselineskip}
\caption{Optimal objective function value of SINR versus SNR, with INR =30~dB, $\rho_1=0.5\tr\hat\bR$, $\rho_2=0.9\tr\hat\bR$ and $T=50$, averaged over 200 simulation runs.}
\label{fig-SINR-value}
%\vspace*{0\baselineskip}
\end{figure}

\begin{figure}[!h]
\centerline{\resizebox{.5\textwidth}{!}{\includegraphics{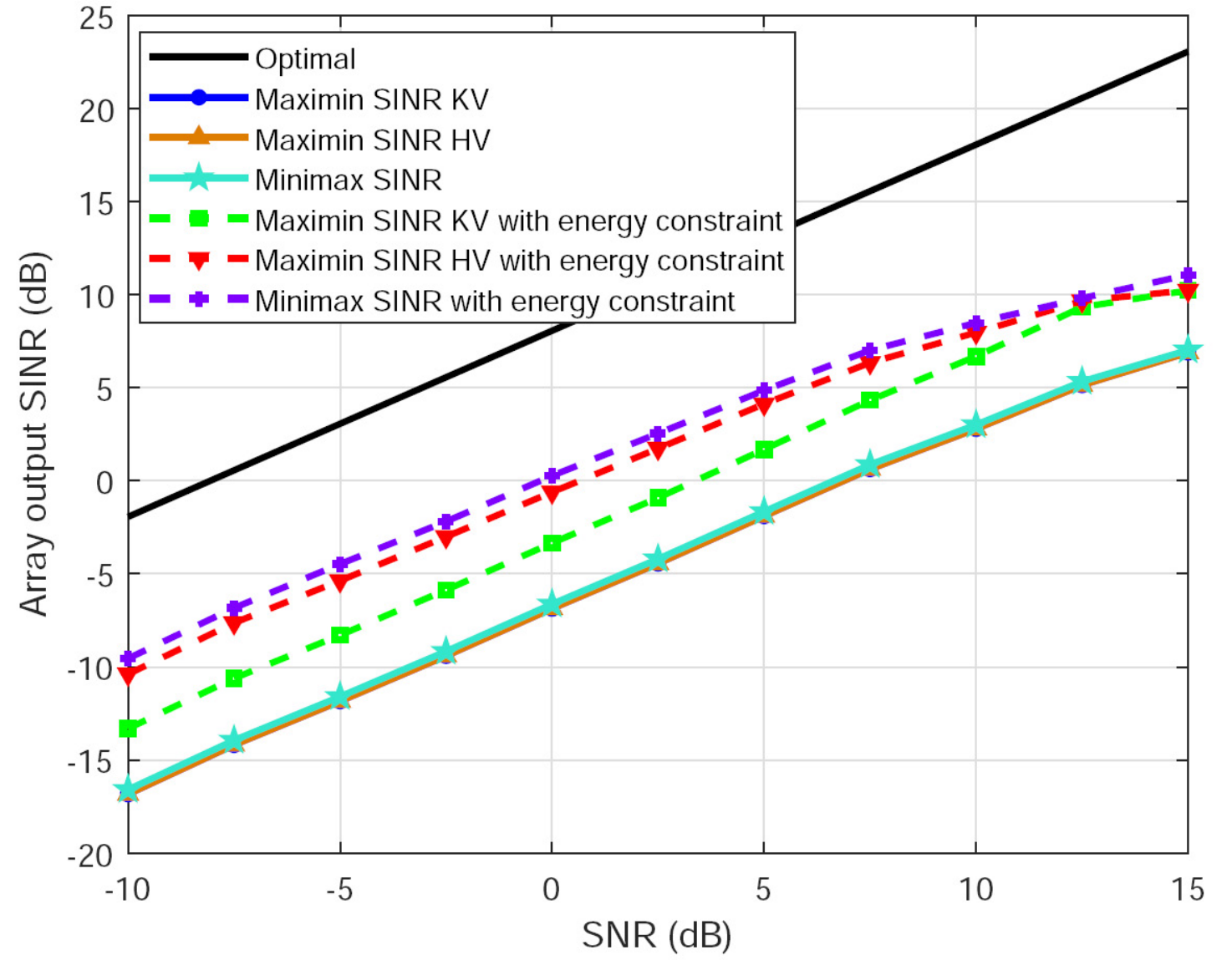}}
}
%\vspace*{-.5\baselineskip}
\caption{Actual array output SINR versus SNR, with INR =30~dB, $\rho_1=0.5\tr\hat\bR$, $\rho_2=0.9\tr\hat\bR$ and $T=50$, averaged over 200 simulation runs.}
\label{fig-array-ouput-SINR}
%\vspace*{0\baselineskip}
\end{figure}

Fig.~\ref{fig-SINR-value} shows the optimal values of the minimax SINR problem with the energy constraint, alongside those of the maximin SINR problem with the energy constraint computed using the approximate algorithms in \cite{KV13-tsp} and \cite{HV18-spl}. It can be seen that the optimal values obtained by the two iterative methods for the maximin SINR problem are very close to each other but are lower than the value of the minimax SINR problem. This indicates that the approximate approaches cannot achieve the global optimum of the maximin SINR problem with the energy constraint. A similar observation holds for the case without the energy constraint: the approximate solutions for the maximin SINR problem in \cite{KV13-tsp} and \cite{HV18-spl} are slightly lower than the corresponding minimax SINR value.

Fig.~\ref{fig-array-ouput-SINR} shows the array output SINR achieved by the beamformers computed from the maximin and minimax SINR problems, both with and without the energy constraint. Three methods are compared: (i)~the minimax SINR problem solved via the equivalent SDP problem \eqref{min-max-2018-general-rank-1-inner-max-dual} or \eqref{min-max-2018-general-rank-1-inner-max-dual-Bprime} in a single short, (ii)~the maximin SINR problem solved by the SDP-based approximation method in \cite{KV13-tsp}, and (iii)~the maximin SINR problem solved by the lighter SOCP-based approximation method in \cite{HV18-spl}.
It is observed that the beamformer obtained from the minimax SINR problem with the energy constraint achieves the highest output SINR. In contrast, the beamformers obtained by solving the maximin SINR problem with the energy constraint using the methods in \cite{KV13-tsp} and \cite{HV18-spl} perform very close to each other but slightly worse than the minimax SINR beamformer. These three performance curves lie above the other three curves corresponding to the minimax and maximin SINR problems without the energy constraint -- namely, the single-short SDP approach and the two iterative approximate algorithms. The relative behaviors of the latter three curves are similar to those of the former three, indicating that the minimax SINR-based beamformers consistently achieve higher array output SINR than the maximin SINR-based beamformers obtained by approximate methods.

\begin{figure}[!h]
\centerline{\resizebox{.51\textwidth}{!}{\includegraphics{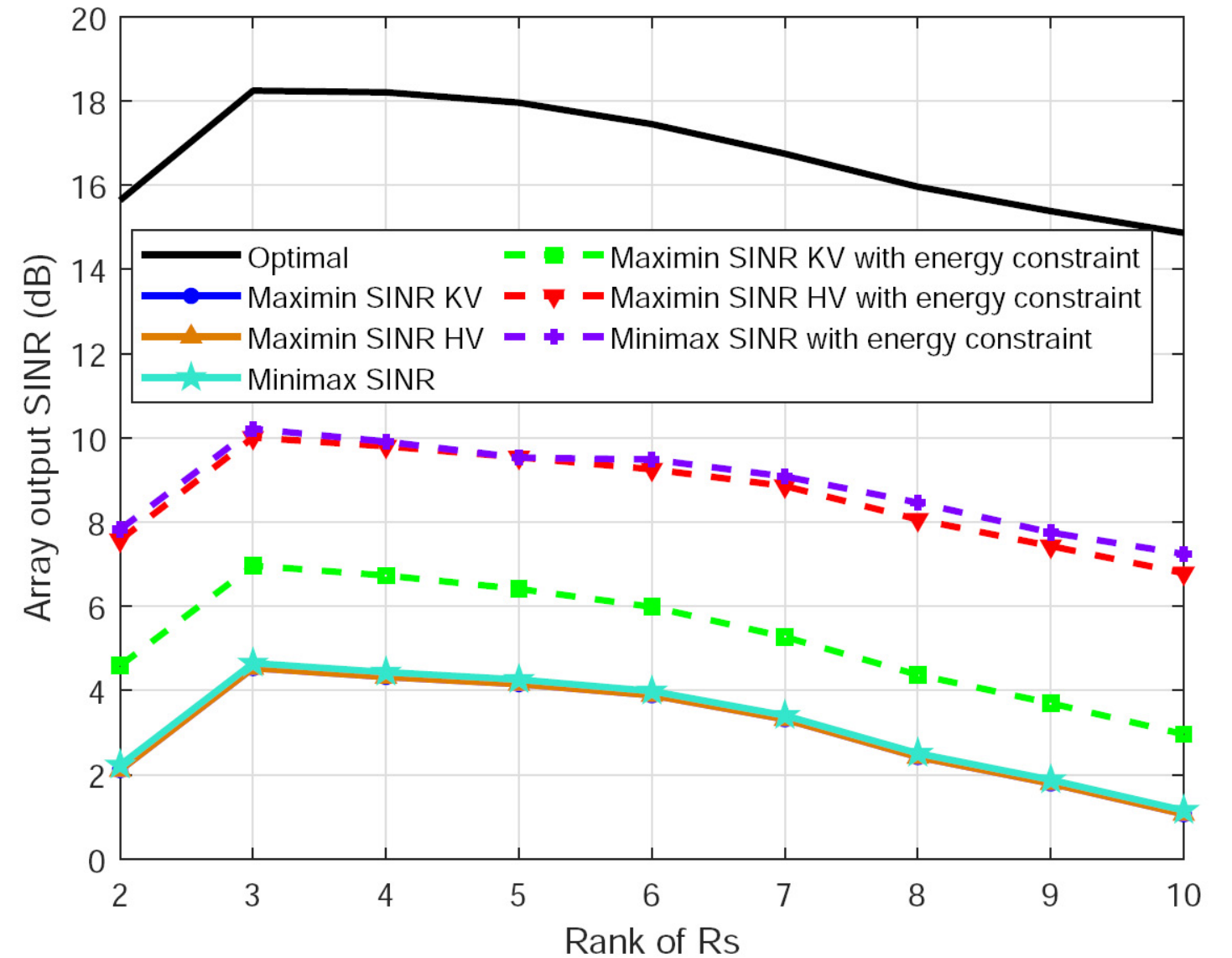}}
}
%\vspace*{-.5\baselineskip}
\caption{The array output SINR versus the rank of actual covariance $\bR_s$, with SNR=10~dB, INR =30~dB, $\rho_1=0.5\tr\hat\bR$, $\rho_2=0.9\tr\hat\bR$ and $T=50$, averaged over 200 simulation runs.}
\label{fig-sinr-rank}
%\vspace*{0\baselineskip}
\end{figure}

We also investigate how the actual angular spread of the desired signal affects beamformer performance, as variations in angular spread change the rank of the desired signal covariance matrix $\bR_s$. Similarly, the number of target signals \cite{Amin-tsp19} can alter the rank of $\bR_s$. All other settings remain the same, with the only difference being the actual angular spread of the desired signal, which determines the rank of $\bR_s$. Specifically, the angular spreads are set to $0.15^\circ, 1^\circ, 2^\circ, 5^\circ, 9^\circ, 14^\circ, 20^\circ, 25^\circ, 30^\circ$, corresponding to $\bR_s$ ranks of $2,3,4,5,6,7,8,9,10$, respectively. The SNR is set to 10~dB.

Fig.~\ref{fig-sinr-rank} plots the beamformer output SINR versus the rank of $\bR_s$. For the minimax SINR problem and the maximin SINR problem (both with the energy constraint), with the former solved via SDP and the latter via two approximate algorithms, we observe that: (i)~the minimax SINR beamformer performs very close to the maximin SINR beamformer solved by the SOCP-based approximation, and (ii)~both outperform the maximin SINR beamformer solved by the SDP-based approximation.
Interestingly, for the minimax and maximin SINR problems without the energy constraint, the three beamformers computed via the corresponding approaches exhibit very similar performance, indicating that the energy constraint primarily influences the performance when using the approximate maximin SINR algorithms.

\section{Conclusion}
We have solved the problem of obtaining a globally optimal solution for the RAB problem via worst-case SINR maximization (the maximin SINR problem) for general-rank signal models. We proved that the maximin and minimax SINR problems are equivalent when the uncertainty sets of the desired signal covariance and the INC matrices are convex and closed. In this case, the minimax SINR problem can be reformulated as a convex problem -- specifically, an SDP when the uncertainty sets are described by finitely many LMI constraints. This relaxes the compactness requirement that was previously (has been known for nearly twenty years) shown to be necessary for proving the equivalence between the minimax SINR and maximin SINR problems in the case of rank-one signal model.

Consequently, a solution to the SDP problem (equivalent to the minimax SINR problem) is also globally optimal for the RAB problem formulated as a maximin SINR problem. In contrast, existing approximation algorithms for the maximin problem return only locally optimal solutions. Our simulations confirm that there can be a positive gap between the globally optimal value obtained via the SDP and the suboptimal values obtained using existing approximate methods. Moreover, the optimal RAB beamformer computed from the SDP achieves a higher array output SINR than the approximate beamformers obtained via iterative maximin algorithms.

\appendix

\subsection{Proof of Proposition \ref{min-max-convex-opt-problem}}\label{proof-prop-min-max-convex-opt-problem}
\proof Let $\lambda$ be the optimal value for the inner maximization problem in \eqref{min-max-2018-general-rank-1}. Hence, we have
\begin{equation}\label{min-max-2018-general-rank-1-inner-max-dual-proof-1}
\lambda=\underset{\bw\ne\bzero}{\sf{maximize}}~~\frac{\bw^H\bQ\bQ^H\bw}{\bw^H\bR_{1}\bw}=\underset{\bu\ne\bzero}{\sf{maximize}}~~\frac{\bu^H\bR_1^{-\frac{1}{2}}\bQ\bQ^H\bR_1^{-\frac{1}{2}}\bu}{\|\bu\|^2},
\end{equation}
with the variable transformation
\begin{equation}\label{variable-transformation-w-u}
\bw=\bR_1^{-\frac{1}{2}}\bu.
\end{equation}

Since the optimal value for the maximization problem in the right-hand side of \eqref{min-max-2018-general-rank-1-inner-max-dual-proof-1} is equal to $\lambda_1(\bR_1^{-\frac{1}{2}}\bQ\bQ^{H}\bR_1^{-\frac{1}{2}})$, which is also the optimal value for the following minimization problem (the dual problem)
\begin{equation}\label{min-max-2018-general-rank-1-inner-max-dual-proof-1-dual}
\begin{array}[c]{cl}
\underset{\lambda}{\sf{minimize}} & \lambda\\
\sf{subject\;to}           & \lambda\bI\succeq \bR_1^{-\frac{1}{2}}\bQ\bQ^{H}\bR_1^{-\frac{1}{2}}. %\left[\begin{array}{cc}\bR_{1}&\bQ\\ \bQ^H&\lambda\bI\end{array}\right]\succeq\bzero\\
%                           & \bQ\in\bar{\cal B}_0,\,\bR_1\in\bar{\cal B}_1,
\end{array}%\right.
\end{equation}
The constraint in problem \eqref{min-max-2018-general-rank-1-inner-max-dual-proof-1-dual} is identical to $\lambda\bR_1\succeq\bQ\bQ^H$, and it follows that \eqref{min-max-2018-general-rank-1-inner-max-dual-proof-1-dual} is equivalent to
\begin{equation}\label{min-max-2018-general-rank-1-inner-max-dual-proof-1-dual-1}
\begin{array}[c]{cl}
\underset{\lambda}{\sf{minimize}} & \lambda\\
\sf{subject\;to}           & \left[\begin{array}{cc}\bR_{1}&\bQ\\ \bQ^H&\lambda\bI\end{array}\right]\succeq\bzero.
%                           & \bQ\in\bar{\cal B}_0,\,\bR_1\in\bar{\cal B}_1,
\end{array}%\right.
\end{equation}
It is verified that the inner maximization problem in \eqref{min-max-2018-general-rank-1} is equivalent to \eqref{min-max-2018-general-rank-1-inner-max-dual-proof-1-dual-1} and therefore, the minimax problem \eqref{min-max-2018-general-rank-1} can be recast into problem \eqref{min-max-2018-general-rank-1-inner-max-dual}, which is a convex problem indeed.

%Suppose that $(\tilde\bQ,\tilde{\bR_1}$ is a feasible point in $\bar{\cal B}_0\times\bar{\cal B}_1$. Then, if we put the constraint $\lambda\in[0,\lambda_1(\tilde{\bR}_1^{-\frac{1}{2}}\tilde{\bQ}\tilde{\bQ}^{H}\tilde{\bR}_1^{-\frac{1}{2}})]$ into problem \eqref{min-max-2018-general-rank-1-inner-max-dual}, the optimal solution set of it will be not altered.
Since any feasible $\lambda$ in \eqref{min-max-2018-general-rank-1-inner-max-dual} is nonnegative, we can rewrite \eqref{min-max-2018-general-rank-1-inner-max-dual} as
\begin{equation}\label{min-max-2018-general-rank-1-inner-max-dual-proof-2}
\begin{array}[c]{cl}
\underset{\lambda,\bQ,\bR_1}{\sf{minimize}} & \lambda^2\\
\sf{subject\;to}           & \left[\begin{array}{cc}\bR_{1}&\bQ\\ \bQ^H&\lambda\bI\end{array}\right]\succeq\bzero\\
                           & \bQ\in\bar{\cal B}_0,\,\bR_1\in\bar{\cal B}_1.
\end{array}%\right.
\end{equation}
It is verified that the feasible set of problem \eqref{min-max-2018-general-rank-1-inner-max-dual-proof-2} is closed and the objective function is coercive, and thus, it follows from the attainment theorem under coerciveness (see, e.g., \cite[Theorem 2.32]{Beck-book-2014}) that the convex problem \eqref{min-max-2018-general-rank-1-inner-max-dual} is solvable.

Suppose that $(\lambda^\star,\bQ^\star,\bR_1^\star)$ is an optimal solution for \eqref{min-max-2018-general-rank-1-inner-max-dual}. Then, we have $\lambda^\star=\lambda_1(\bR_1^{\star -\frac{1}{2}}\bQ^\star\bQ^{\star H}\bR_1^{\star -\frac{1}{2}})$. Let $\bu_1^\star$ be an eigenvector associated with $\lambda^\star$ and $\bw^\star=\bR_1^{\star -\frac{1}{2}}\bu_1^\star$. From \eqref{variable-transformation-w-u}, we conclude that $(\bQ^\star,\bR^
\star,\bw^\star)$ is optimal for the minimax problem \eqref{min-max-2018-general-rank-1}. The proof is complete.
\endproof

\subsection{Proof of Lemma \ref{lambda-1-cvx-fcn}}\label{proof-lemma-lambda-1-cvx-fcn}
\proof Showing the convexity of $f(\bQ,\bR_1)$ is equivalent to checking that the epigraph $\{(\bQ,\bR_1,t)~|~f(\bQ,\bR_1)\le t\}$ is a convex set. In fact, according to \eqref{the-objective-1}, we can verify that the epigraph is given as
\begin{equation}\label{epigraph-1}
\left\{(\bQ,\bR_1,t)~|~\left[\begin{array}{cc}\bR_1&\bQ\\ \bQ^H&t\bI\end{array}\right]\succeq\bzero\right\},
\end{equation}
which is a set of LMIs. Therefore, the epigraph is convex.
\endproof

\subsection{Proof of Proposition \ref{first-partial-derivative-1}}\label{proof-prop-first-partial-derivative-1}
\proof It follows from \cite{eigen-derivative} that
\begin{equation}\label{first-partial-derivative-val-2}
\frac{\partial f}{\partial \bQ}=\frac{\bu_1^H\partial(\bR_1^{-\frac{1}{2}}\bQ\bQ^{H}\bR_1^{-\frac{1}{2}})\bu_1}{\partial\bQ}=\frac{\partial (\bu_1^H\bR_1^{-\frac{1}{2}}\bQ\bQ^{H}\bR_1^{-\frac{1}{2}}\bu_1)}{\partial\bQ},
\end{equation}where $\bu_1$ is an eigenvector associated with $\lambda_1(\bR_1^{-\frac{1}{2}}\bQ\bQ^{H}\bR_1^{-\frac{1}{2}})$. Therefore, we have
\begin{eqnarray}\nonumber
\frac{\partial (\bu_1^H\bR_1^{-\frac{1}{2}}\bQ\bQ^{H}\bR_1^{-\frac{1}{2}}\bu_1)}{\partial\bQ}&=&\frac{\partial \tr((\bR_1^{-\frac{1}{2}}\bu_1\bu_1^H\bR_1^{-\frac{1}{2}}\bQ\bQ^{H}))}{\partial\bQ}\\ \label{first-partial-derivative-val-3}
 &=&2\bR_1^{-\frac{1}{2}}\bu_1\bu_1^H\bR_1^{-\frac{1}{2}}\bQ.
\end{eqnarray}
Then, \eqref{first-partial-derivative-val-2} and \eqref{first-partial-derivative-val-3} lead to
\begin{equation}\label{first-partial-derivative-val-4}
\frac{\partial f}{\partial \bQ}=2\bR_1^{-\frac{1}{2}}\bu_1\bu_1^H\bR_1^{-\frac{1}{2}}\bQ,
\end{equation}which implies that
\begin{equation}\label{first-partial-derivative-val-5}
\frac{\partial f}{\partial \bQ}\left|_{(\bQ^\star,\bR_1^\star)}\right.=2\bR_1^{\star-\frac{1}{2}}\bu_1^\star\bu_1^{\star H}\bR_1^{\star-\frac{1}{2}}\bQ^\star=2\bw^\star\bw^{\star H}\bQ^\star.
\end{equation}Based on \eqref{first-partial-derivative-val-5}, it is not hard to verify \eqref{first-partial-derivative-val-11}.
\endproof

\subsection{Proof of Lemma \ref{svd-u1-v1-lemma}}\label{proof-lemma-svd-u1-v1-lemma}
\proof It is easy to verify that
\begin{eqnarray}
\bu_1^H\bA\bv_1&=&\sqrt{\lambda_1}\\
               &=&\sqrt{\lambda_1(\bA\bA^H)}=\sqrt{\lambda_1(\bA^H\bA)}\\
               &=&\sqrt{\bu_1^H\bA\bA^H\bu_1}=\sqrt{\bv_1^H\bA^H\bA\bv_1}.
\end{eqnarray}
Let $\hat\bv_1=\bA\bv_1$. Then, we have
\begin{equation}\label{svd-u1-v1-proof-1}
\bu_1^H\hat\bv_1=\sqrt{\lambda_1}=\|\hat\bv_1\|,
\end{equation}and hence
\begin{equation}\label{svd-u1-v1-proof-2}
\bu_1^H\hat\bv_1\hat\bv_1^H\bu_1=\lambda_1=\|\hat\bv_1\|^2=\lambda_1(\hat\bv_1\hat\bv_1^H).
\end{equation}
Recall that
\begin{equation}\label{svd-u1-v1-proof-3}
\bu_1^H\hat\bv_1\hat\bv_1^H\bu_1=\lambda_1(\hat\bv_1\hat\bv_1^H)=\underset{\|\bu\|^2=1}{\sf{maximize}}~~{\bu^H\hat\bv_1\hat\bv_1^H\bu}.
\end{equation}
Therefore, $\bu_1$ is optimal for the maximization problem in the right-hand side of \eqref{svd-u1-v1-proof-3}, and it follows from \eqref{svd-u1-v1-proof-2} that $\lambda_1$ is the largest eigenvalue of $\hat\bv_1\hat\bv_1^H$ and $\bu_1$ is the eigenvector associated with $\lambda_1$. Since $\hat\bv_1\hat\bv_1^H$ is of rank one, then we have
\begin{equation}\label{svd-u1-v1-proof-4}
\hat\bv_1\hat\bv_1^H=\lambda_1\bu_1\bu_1^H,
\end{equation}namely,
\begin{equation}\label{svd-u1-v1-proof-5}
\bA\bv_1\bv_1^H\bA^H=\lambda_1\bu_1\bu_1^H.
\end{equation}This shows equality \eqref{svd-u1-v1-1}.
Similarly, we can prove \eqref{svd-u1-v1-2}. The proof is complete.
\endproof

\subsection{Proof of Proposition \ref{second-partial-derivative-1}}\label{proof-prop-second-partial-derivative-1}
\proof Let $\bR_1^{-\frac{1}{2}}\bQ=\bU\bSigma\bV^H$ be the singular value decomposition, and $\bu_1$ and $\bv_1$ be the first column vectors of $\bU$ and $\bV$, respectively. Observe that
\begin{equation}\label{the-objective-2}
f(\bQ,\bR_1)=\lambda_1(\bR_1^{-\frac{1}{2}}\bQ\bQ^{H}\bR_1^{-\frac{1}{2}})=\lambda_1(\bQ^{H}\bR_1^{-1}\bQ).
\end{equation}
Therefore, it follows from \cite{eigen-derivative} that
\begin{equation}\label{second-partial-derivative-val-2}
\frac{\partial f}{\partial \bR_1}=\frac{\bv_1^H\partial(\bQ^{H}\bR_1^{-1}\bQ)\bv_1}{\partial\bR_1}=\frac{\partial (\bv_1^H\bQ^{H}\bR_1^{-1}\bQ\bv_1)}{\partial\bR_1},
\end{equation}
which is further equal to
\begin{equation}\label{second-partial-derivative-val-3}
-(\bR_1^{-1}\bQ\bv_1\bv_1^H\bQ^{H}\bR_1^{-1})=-(\bR_1^{-\frac{1}{2}}\bR_1^{-\frac{1}{2}}\bQ\bv_1\bv_1^H\bQ^{H}\bR_1^{-\frac{1}{2}}\bR_1^{-\frac{1}{2}}).
\end{equation}
Due to \eqref{svd-u1-v1-1}, the latter in turn is equal to
\begin{equation}\label{second-partial-derivative-val-4}
-\lambda_1\bR_1^{-\frac{1}{2}}\bu_1\bu_1^H\bR_1^{-\frac{1}{2}}.
\end{equation}
In \eqref{second-partial-derivative-val-4}, $\lambda_1=\lambda_1(\bR_1^{-\frac{1}{2}}\bQ\bQ^{H}\bR_1^{-\frac{1}{2}})$. Therefore, we obtain
\begin{equation}\label{second-partial-derivative-val-5}
\frac{\partial f}{\partial \bR_1}\left|_{(\bQ^\star,\bR_1^\star)}\right.=-\lambda^\star\bR_1^{\star -\frac{1}{2}}\bu_1^\star\bu_1^{\star H}\bR_1^{\star -\frac{1}{2}}=-\lambda^\star\bw^\star\bw^{\star H}.
\end{equation}
Using the latter, we immediately obtain \eqref{second-partial-derivative-val-11}.
\endproof

\subsection{Proof of Proposition \ref{cvx-obj-fcn-prop}}\label{proof-prop-cvx-obj-fcn-prop}
\proof Assume that $(\bQ_i,\bR_{1i})\in\mathbb{C}^{N\times M}\times{\cal H}_+^N$, $i=1,2$, and $\alpha\in(0,1)$. We wish to show that
\begin{equation}\label{cvx-fcn}
\begin{array}{l}
\alpha h(\bQ_1,\bR_{11}) + (1-\alpha) h(\bQ_2,\bR_{12})\\
~~~~~~~~~~~~~~~~~~~~~\ge h(\alpha\bQ_1+(1-\alpha)\bQ_2,\alpha\bR_{11}+(1-\alpha)\bR_{12}),
\end{array}
\end{equation}i.e.,
\begin{equation}\label{cvx-fcn-1}
\begin{array}{l}
\frac{\alpha\|\bQ_1^H\bw\|^2}{\bw^{H}\bR_{11}\bw}+\frac{(1-\alpha)\|\bQ_2^H\bw\|^2}{\bw^{H}\bR_{12}\bw}\ge\frac{\|(\alpha\bQ_1+(1-\alpha)\bQ_2)^H\bw\|^2}{\bw^{H}(\alpha\bR_{11}+(1-\alpha)\bR_{12})\bw},
%~~~~~~~~~~~~~~~~~~~~~~~~~~~~~~~~~~~~~~~~~~\ge\frac{\|(\alpha\bQ_1+(1-\alpha)\bQ_2)\bw^\star\|^2}{\bw^{\star H}(\alpha\bR_{11}+(1-\alpha)\bR_{12})\bw^\star}.
\end{array}
\end{equation}which is equivalent to
\begin{equation}\label{cvx-fcn-2}
\begin{array}{l}
\bw^{H}(\alpha\bR_{11}+(1-\alpha)\bR_{12})\bw(\alpha\|\bQ_1^H\bw\|^2\bw^{H}\bR_{12}\bw\\
~~~~~~~~~~~~~~~~~~~~~~~~~~~~~~~~~~~~~+(1-\alpha)\|\bQ_2^H\bw\|^2\bw^{H}\bR_{11}\bw)\\
~~~~~~~~~\ge\|(\alpha\bQ_1+(1-\alpha)\bQ_2)^H\bw\|^2\bw^{H}\bR_{11}\bw\bw^{H}\bR_{12}\bw.
%~~~~~~~~~~~~~~~~~~~~~~~~~~~~~~~~~~~~~~~~~~\ge\frac{\|(\alpha\bQ_1+(1-\alpha)\bQ_2)\bw^\star\|^2}{\bw^{\star H}(\alpha\bR_{11}+(1-\alpha)\bR_{12})\bw^\star}.
\end{array}
\end{equation}
To prove \eqref{cvx-fcn-2}, let us recast the left-hand side of the inequality into
\begin{equation}\label{cvx-fcn-2-LHS}
\begin{array}{l}
(\alpha\bw^{H}\bR_{11}\bw+(1-\alpha)\bw^{H}\bR_{12}\bw)(\|\alpha\bQ_1^H\bw\|^2\frac{\bw^{H}\bR_{12}\bw}{\alpha}\\
~~~~~~~~~~~~~~~~~~~~~~~~~~~~~~~~~~~+\|(1-\alpha)\bQ_2^H\bw\|^2\frac{\bw^{H}\bR_{11}\bw}{(1-\alpha)}).
%~~~~~~~~~~~~~~~~~~~~~~~~~~~~~~~~~~~~~~~~~~\ge\frac{\|(\alpha\bQ_1+(1-\alpha)\bQ_2)\bw^\star\|^2}{\bw^{\star H}(\alpha\bR_{11}+(1-\alpha)\bR_{12})\bw^\star}.
\end{array}
\end{equation}
It follows from the Cauchy-Schwarz inequality that \eqref{cvx-fcn-2-LHS} is greater than or equal to
\begin{equation}\label{cvx-fcn-2-RHS}
\begin{array}{l}
((\|\alpha\bQ_1^H\bw\|+\|(1-\alpha)\bQ_2^H\bw\|)\sqrt{\bw^{H}\bR_{11}\bw \bw^{H}\bR_{12}\bw})^2,
%~~~~~~~~~~~~~~~~~~~~~~~~~~~~~~~~~~~~~~~~~~\ge\frac{\|(\alpha\bQ_1+(1-\alpha)\bQ_2)\bw^\star\|^2}{\bw^{\star H}(\alpha\bR_{11}+(1-\alpha)\bR_{12})\bw^\star}.
\end{array}
\end{equation}
which, by the triangle inequality, is greater than or equal to
\begin{equation}\label{cvx-fcn-2-RHS-2}
\begin{array}{l}
\|(\alpha\bQ_1+(1-\alpha)\bQ_2)^H\bw\|^2\bw^{H}\bR_{11}\bw \bw^{H}\bR_{12}\bw.
%~~~~~~~~~~~~~~~~~~~~~~~~~~~~~~~~~~~~~~~~~~\ge\frac{\|(\alpha\bQ_1+(1-\alpha)\bQ_2)\bw^\star\|^2}{\bw^{\star H}(\alpha\bR_{11}+(1-\alpha)\bR_{12})\bw^\star}.
\end{array}
\end{equation}
%coincides exactly with the right-hand side of \eqref{cvx-fcn-2}, and
Therefore, the inequality \eqref{cvx-fcn-2} holds, and we conclude that $h(\bQ,\bR_1)$ is convex over $\mathbb{C}^{N\times M}\times{\cal H}_+^N$.
\endproof

\subsection{Proof of Proposition \ref{cvx-problem-1-optsoln}}\label{proof-prop-cvx-problem-1-optsoln}
\proof Since $(\bQ^\star,\bR_1^\star)$ is optimal for \eqref{min-max-2018-general-rank-1-inner-max-dual-equv-1} and $\bw^\star$ is defined as in \eqref{define-w-star}, it follows from Proposition~\ref{optimality-conditions-thm} that $(\bw^\star,\bQ^\star,\bR_1^\star)$ satisfies the optimality condition \eqref{optimality-conditions-2}.

On the other hand, problem \eqref{cvx-problem-1} is convex, and therefore, applying \cite[Theorem 9.7]{Beck-book-2014} yields that
\begin{equation}\label{optimality-conditions-min-1}
\Re(\tr(g_1^{\star H}(\bQ-\bQ^\star))+\tr(g_2^{\star H}(\bR_1-\bR_1^\star)))\ge0,\forall\bQ\in\bar{\cal
B}_0,\bR_1\in\bar{\cal B}_1,
\end{equation}
where
\begin{equation}\label{two-partial-derivatives-h-fcn}
g_1^\star=\frac{\partial h}{\partial \bQ}\left|_{(\bQ^\star,\bR_1^\star)}\right.,\,g_2^\star=\frac{\partial h}{\partial \bR_1}\left|_{(\bQ^\star,\bR_1^\star)}\right.,
\end{equation}
and the function $h$ is specified to the objective function of problem \eqref{cvx-problem-1}:
\begin{equation}\label{obj-fcn-cvx-problem-1}
h(\bQ,\bR_1)=\frac{\bw^{\star H}\bQ\bQ^H\bw^\star}{\bw^{\star H}\bR_{1}\bw^\star}.
\end{equation}
It is easy to verify that
\begin{equation}\label{first-partial-derivatives-h-fcn}
\frac{\partial h}{\partial \bQ}=\frac{\partial\tr\left(\frac{\bw^\star\bw^{\star H}\bQ\bQ^H}{\bw^{\star H}\bR_1\bw^\star}\right)}{\partial\bQ}=2\frac{\bw^\star\bw^{\star H}\bQ}{\bw^{\star H}\bR_1\bw^\star},
\end{equation}
and
\begin{equation}\label{second-partial-derivatives-h-fcn}
\frac{\partial h}{\partial \bR_1}=\frac{\partial\tr\left(\frac{\bw^\star\bw^{\star H}\bQ\bQ^H}{\bw^{\star H}\bR_1\bw^\star}\right)}{\partial\bR_1}=-\frac{\bw^{\star H}\bQ\bQ^H\bw^\star}{(\bw^{\star H}\bR_1\bw^\star)^2}\bw^\star\bw^{\star H}.
\end{equation}
Using \eqref{two-partial-derivatives-h-fcn}, we have
\begin{equation}\label{first-partial-derivatives-h-fcn-1}
g_1^\star=2\bw^\star\bw^{\star H}\bQ^\star,
\end{equation}
where we employ $\bw^{\star H}\bR_1^\star\bw^\star=1$, and
\begin{equation}\label{second-partial-derivatives-h-fcn-1}
g_2^\star=-\lambda^\star\bw^\star\bw^{\star H},
\end{equation}
and apply $\lambda^\star=\bw^{\star H}\bQ^\star\bQ^{\star H}\bw^\star$ (see \eqref{w-Q-QH-w-all-star}).

Therefore, the optimality condition \eqref{optimality-conditions-min-1} can be rewritten as
\begin{equation}\label{optimality-conditions-min-2}
2\Re(\bw^{\star H}(\bQ-\bQ^\star)\bQ^{\star H}\bw^\star)-\lambda^\star\bw^{\star H}(\bR_1-\bR_1^\star)\bw^\star\ge0,
\end{equation}$\forall\bQ\in\bar{\cal B}_0,\bR_1\in\bar{\cal B}_1$, which  is equivalent to
\begin{equation}\label{optimality-conditions-min-3}
2\Re(\bw^{\star H}\bQ\bQ^{\star H}\bw^\star)-\bw^{\star H}(\lambda^\star(\bR_1-\bR_1^\star)+2\bQ^\star\bQ^{\star H})\bw^\star\ge0,
\end{equation}$\forall\bQ\in\bar{\cal B}_0,\bR_1\in\bar{\cal B}_1$. Note that the optimality condition \eqref{optimality-conditions-min-3} for problem \eqref{cvx-problem-1} is the same as \eqref{optimality-conditions-2}. In other words, the two convex problems \eqref{cvx-problem-1} and \eqref{min-max-2018-general-rank-1-inner-max-dual-equv-1} share the same optimality condition. Hence, $(\bQ^\star,\bR_1^\star)$, which solves \eqref{min-max-2018-general-rank-1-inner-max-dual-equv-1}, is optimal for \eqref{cvx-problem-1} with the optimal value
\begin{equation}\label{opt-val-cvx-problem-1}
h(\bQ^\star,\bR_1^\star)=\frac{\bw^{\star H}\bQ^\star\bQ^{\star H}\bw^\star}{\bw^{\star H}\bR_{1}^\star\bw^\star}=\bw^{\star H}\bQ^\star\bQ^{\star H}\bw^\star,
\end{equation}
and vice versa. The proof is complete.
\endproof

\end{document}